\documentclass[twocolumn,amsthm]{autart}    
\usepackage{cite}
\usepackage{amsmath,amssymb,amsfonts}
\usepackage{algorithmic}
\usepackage{graphicx}
\usepackage{algorithm,algorithmic}
\usepackage{hyperref}
\hypersetup{hidelinks=true}
\usepackage{textcomp}
\usepackage{bm}
\usepackage{hyperref}
\usepackage{cuted}
\usepackage{amsthm}
\usepackage{mathtools}
\usepackage{tensor} 
\usepackage{pifont}
\usepackage{rotating}
\usepackage{transparent}
\usepackage{nicefrac}
\usepackage{xcolor}
\usepackage{xfrac}
\usepackage{layouts}
\usepackage{subcaption}

\DeclareInstance{xfrac}{mathdefault}{math}
{
	numerator-top-sep = 0pt,
	scale-factor = 1,
	scaling = true
}
\usepackage{nicematrix}
\usepackage{tikz}
\usepackage{tikz-cd}
\usepackage{enumitem}

\usetikzlibrary{positioning, arrows.meta, calc, shapes.geometric, backgrounds, fit, decorations.pathreplacing}

\usepackage{pgfplots}
\pgfplotsset{compat=1.18}
\usepgfplotslibrary{fillbetween}

\usepackage{stackrel}
\tikzcdset{every label/.append style = {font = \footnotesize}}

\usepackage{thmtools}
\usepackage{mathtools}
\usepackage{tensor} 

\usepackage{transparent}
\usepackage{color}

\usepackage{tikz}
\usepackage{tikz-cd}
\usepackage{booktabs}
\usepackage{caption}
\usepackage{subfiles}

\usepackage{listings}
\usepackage{framed}

\usepackage[most]{tcolorbox}

\renewcommand{\phi}{\varphi}

\usepackage{stackengine,scalerel,graphicx}
\stackMath

\newcommand{\pushright}[1]{\ifmeasuring@#1\else\omit\hfill$\displaystyle#1$\fi\ignorespaces}

\usepackage{tikz-cd}

\renewcommand{\theenumi}{\roman{enumi}}

\renewcommand{\theenumii}{\alph{enumii}}

\newcommand{\vu}{v_u}

\newcommand{\bb}{\bm{b}}

\newcommand{\AAA}{\bm{A}}

\newcommand{\stil}{\tilde{s}}

\newcommand{\xx}{\bm{x}}

\newcommand{\xb}{\bar{x}}

\newtheorem{problem}{Problem}

\newcommand{\s}{^*}

\usepackage{tabstackengine}
\usepackage{nicematrix}
\usepackage{soul}

\usepackage{pgfplots}
\pgfplotsset{compat=1.18}

\renewcommand{\leq}{\leqslant}
\renewcommand{\geq}{\geqslant}

\makeatletter
\renewenvironment{proof}[1][\proofname]{%
	\par\pushQED{\qed}\normalfont%
	\topsep6\p@\@plus6\p@\relax
	\trivlist\item[\hskip\labelsep\bfseries#1\@addpunct{.}]%
	\ignorespaces
}{%
	\popQED\endtrivlist\@endpefalse
}	
	
\makeatother

\theoremstyle{plain}
\newtheorem{theorem}{Theorem}
\newtheorem{corollary}{Corollary}
\newtheorem{lemma}{Lemma}
\newtheorem{proposition}{Proposition}

\theoremstyle{definition}
\newtheorem{definition}{Definition}

\makeatletter
\renewcommand{\@thebibliography}[1]{\@bibliosize
	\list{\@biblabel{\arabic{enumiv}}}{\settowidth\labelwidth{\@biblabel{#1}}
		\if@nameyear
		\labelwidth\z@ \labelsep\z@ \leftmargin\parindent
		\itemindent-\parindent
		\else
		\labelsep 3\p@ \itemindent\z@
		\leftmargin\labelwidth \advance\leftmargin\labelsep
		\fi
		\itemsep 0pt plus 0.1\@bls
		\usecounter{enumiv}\let\p@enumiv\@empty
		\def\theenumiv{\arabic{enumiv}}}%
	\tolerance\@M
	\hyphenpenalty\@M
	\hbadness5000 \sfcode`\.=1000\relax}
\makeatother

\begin{document}
	
	\begin{frontmatter}
		
		\title{
			Time-Optimal Control of One-Mode Flexible Structures: Analytical Solution and the Cost of Flexibility
%
%
%
%
%
		} 
		
		\thanks[footnoteinfo]{This paper was not presented at any IFAC 
			meeting. 
			}
		
		
		\author[dlr,tue]{Manuel Keppler}\ead{manuel.keppler@dlr.de}
		\address[dlr]{German Aerospace Center (DLR), Wessling, Germany\\[-1.5em]}
		\address[tue]{Eindhoven University of Technology (TU/e), Eindhoven, The Netherlands\\[-1.0em]}

\begin{abstract}
The time-optimal solution for a double integrator is a foundational result: its time-displacement law $T_r= 2\sqrt{L}$ for a normalized step $L$ is structure that numerical methods cannot provide.
We solve the time-optimal control for rest-to-rest maneuvers of a double integrator coupled to a harmonic oscillator, deriving the first analytical solution and closed-form time-displacement law. Synthesis reduces to a single scalar inversion.
This is the model to which the two-mass-spring system, the one-bending-mode flexible structure, and the linearized overhead crane reduce.
A Pythagorean identity $T^2 = T_r^2 + 2T_s^2$ decomposes the optimal maneuver time into the
rigid-body minimum and a synchronization time $T_s \leq \pi$, the
cost of flexibility, giving the sharp envelope $2\sqrt{L} \leq T \leq
2\sqrt{L + \pi^2/2}$. The penalty is scale-dependent. For small
maneuvers $T \propto L^{1/4}$: one oscillator mode costs as much as two additional integrators. For large ones, flexibility is asymptotically free. It vanishes at the natural motions, where the oscillator completes whole cycles. A closed-form sensitivity formula proves the maneuver time non-monotone in stiffness. Adding flexibility to a rigid body never shortens a maneuver, yet softening an already flexible structure can: stiffer is not always faster. The natural motions are robust design targets since first-order sensitivity to stiffness vanishes there.

\end{abstract} 

	\end{frontmatter}

	\section{Introduction}\label{sec:intro}


The time-optimal analytical solution for the double integrator is
a foundational result of optimal control
theory~\cite{mcdonald1950,hopkin1951,bushaw1952}. Its lasting
value lies in the qualitative structure it reveals: bang-bang
controls with a single switch, and the time-displacement law
$T_{r} = 2\sqrt{L}$ relating maneuver time to step size, structure
that numerical methods cannot provide. Yet for the natural next
case, a double integrator coupled to a harmonic oscillator through
a shared bounded input, 
the switching times have not been available in closed form, and with them no time-displacement law.


The same fourth-order system underlies three classical benchmarks: the two-mass-spring system~\cite{wie1992}, the one-bending-mode flexible structure, and the linearized overhead crane (Fig.~\ref{fig:unification}). These arise across engineering domains: drive-train compliance in robotics~\cite{hirzinger2001,pratt1995} and precision mechatronics~\cite{altintas2011,hori1999}; flexible spacecraft with elastic appendages~\cite{williams1987,junkins1993} and proximity operations under the Hill--Clohessy--Wiltshire model~\cite{bevilacqua2010,sevier2024}; and linearized overhead cranes with suspended payloads~\cite{sakawa1982,auernig1987}. Each reduces, via state transformation and time scaling (App.~\ref{app:canonical}), to a common canonical model, the lowest-order system with both a rigid-body and a flexible mode.

Closed-form time-optimal solutions exist for the double and triple integrator~\cite{bushaw1952,lee1967,akulenko2000} and the harmonic oscillator~\cite{pontryagin1962,athans2006}, and recent work continues to extend the analytically tractable cases~\cite{he2020,wang2025} and to reduce
transfers between arbitrary states to the no-rest-to-rest problem~\cite{romano2020}. Yet the coupled fourth-order system has resisted analytical synthesis. The difficulty is structural: under Pontryagin's minimum principle, the oscillator's costates introduce trigonometric terms into the switching function, turning the boundary-value problem into a system of transcendental equations in the switching times. Recent work has relied on hybrid analytical-numerical methods~\cite{sevier2024} or has studied the maneuver's robustness under process noise~\cite{bhattacharjee2025}, leaving the exact rest-to-rest solution open.

Although structural properties of the optimal control are known, such as phase-space symmetries~\cite{barbieri1988}, odd symmetry of the optimal control~\cite{singh1989}, a three-switch upper bound on the bang-bang structure~\cite{pao1990,pao1996}, and a state-feedback switching-hypersurface
characterization~\cite{barbieri1993}, the exact mappings between step size, maneuver time, and switching times have remained elusive. Pao~\cite{pao1996} observed numerically that isolated step sizes admit a single switch; \cite{keppler2020} named these \emph{natural motions}, characterized them exactly, and constructed candidate optimal trajectories of the form derived here, without establishing their optimality.

More broadly, the rigid-body time-displacement law has stood for seventy years \cite{bushaw1952,athans2006} without a flexible counterpart. 
%
A heuristic small-step approximation exists for the flexible-spacecraft model~\cite{ben-asher1987}, but the exact penalty across arbitrary steps has remained unknown.
A fundamental question thus stands open, even at the simplest mechanical
model:
\emph{what is the exact cost of flexibility in time-optimal control?}

Answering it requires the analytical solution, long impeded by the transcendental boundary-value problem of Pontryagin's minimum principle. We obtain it by working in the phase plane. The main contributions are:


\smallskip
\begin{enumerate}
	\item We derive the first analytical solution for the 
	time-optimal control for rest-to-rest maneuvers
	of the canonical model of Fig.~\ref{fig:unification} (Theorem~\ref{thm:analytical_sol}), reducing synthesis to a single scalar inversion $\alpha = w^{-1}(L)$.
	
	\item We establish the Pythagorean identity
	$T^{2} = T_{r}^{2} + 2T_{s}^{2}$
	(Theorem~\ref{thm:pythagorean}), which decomposes the optimal time
	into the rigid-body minimum $T_{r} = 2\sqrt{L}$ and a
	synchronization cost $T_{s} \leqslant \pi$, a closed-form
	time-displacement law (Corollary~\ref{cor:kepler}), and the sharp
	envelope $2\sqrt{L} \leqslant T \leqslant 2\sqrt{L+\pi^{2}/2}$
	(Corollary~\ref{cor:bounds}).



	\item We prove the asymptotic scaling
	(Proposition~\ref{prop:scaling}): $T \propto L^{1/4}$ as
	$L \to 0$, so one oscillator mode penalizes short maneuvers as
	severely as two additional integrators in the chain
	($T/T_{r} \to \infty$). For large maneuvers, flexibility is
	asymptotically free ($T/T_{r} \to 1$).
	
	\item A closed-form sensitivity formula proves the
	maneuver time non-monotone in stiffness, with a minimum at each
	natural motion $L=(2n\pi)^{2}$, where the penalty vanishes and the
	rigid-body minimum is attained (Theorem~\ref{thm:sign_reversal},
	Corollary~\ref{cor:natural_anti}). So adding flexibility to a rigid body never shortens a maneuver, yet softening an already flexible structure can.
	Section~\ref{sec:design} translates these results into design guidelines in the native parameters of the three models in Fig.~\ref{fig:unification}.
%
%
\end{enumerate}		
		\section{Problem Formulation}
	\label{sec:problem}
	We consider the time-optimal control for rest-to-rest maneuvers of the fourth-order canonical model
	\begin{gather}\label{eq:x_dyn}
			\dot{\mathbf{x}}(\tau) = \mathbf{A}\mathbf{x}(\tau) + \mathbf{b} u(\tau), \quad |u(\tau)|\leq 1,
			\\[0.4em]
			\AAA=\renewcommand{\arraystretch}{1.0}
			\begin{bNiceMatrix}[columns-width=0.8em,margin=0.4em]
				0&1&0&0\\
				0&0&0&0\\
				0&0&0&1\\
				0&0&-1&0\\
			\end{bNiceMatrix},
			\quad
			\bb=
			\begin{bNiceMatrix}[margin]
				0\\1\\0\\1
			\end{bNiceMatrix}.
	\end{gather}
	%
%
	Here $\tau$, $u$, and $\mathbf{x}$ are dimensionless: time is
	scaled by the oscillator period, the input by the actuator
	bound, and the state by the resulting units, so that the
	frequency and input limit are unity.
	The three models of Fig.~\ref{fig:unification} reduce to this form (Appendix~\ref{app:canonical}). $\xx$ decomposes into a rigid-body pair
	$(x_1, x_2)$ and a harmonic-oscillator pair
	$(x_3, x_4)$, both driven by the common
	input. In the two-mass realization $x_1$ is the
	center of mass and $x_3$ the spring deflection 
	%
	\begin{equation*}
		\frac{d}{d\tau}
		\renewcommand{\arraystretch}{1.0}
		\begin{bNiceMatrix}[margin=0.1em]
			x_1 \\ x_2 
		\end{bNiceMatrix}
		= \renewcommand{\arraystretch}{1.0}
		\begin{bNiceMatrix}[margin=0.1em]
			x_2 \\ u 
		\end{bNiceMatrix},
		\qquad
		\frac{d}{d\tau}
		\renewcommand{\arraystretch}{1.0}
		\begin{bNiceMatrix}[margin=0.1em] 
			x_3 \\ x_4 \end{bNiceMatrix}
		= \renewcommand{\arraystretch}{1.0}
		\begin{bNiceMatrix}[margin=0.1em] 
			x_4 \\ -x_3 + u \end{bNiceMatrix}.
	\end{equation*}
	%

	\begin{problem}\label{prob:top}\label{prob:TOC}
		Find the control
		$u\colon [-T/2, T/2] \to [-1, 1]$
		that drives~\eqref{eq:x_dyn} from
		$\mathbf{x}(-T/2) = (\bar{x}_1, 0, 0, 0)^{\!\top}$,
		$\bar{x}_1 \in \mathbb{R}$, to
		$\mathbf{x}(T/2) = \mathbf{0}$ in minimum
		time~$T$.
	\end{problem}
	
	The magnitude $L := |\bar{x}_1|$ is the \emph{step size} and $T$ the \emph{maneuver time}. The centered interval $[-T/2, T/2]$ exposes the odd symmetry below.
	The system is controllable and normal~\cite{athans2006}, so the unique time-optimal control is bang-bang with finitely many switches.
	
	%
	\smallskip
	
	\begin{description}
		\item[\normalfont (P1)] Any time-optimal trajectory is
		symmetric about the line
		$x_1 = \bar{x}_1/2$ in the $x_1 x_2$-plane
		and about the $x_4$-axis in the
		$x_3 x_4$-plane~\cite{barbieri1988}.
		\item[\normalfont(P2)] The optimal control $u^*$ is odd:
		$u^*(\tau) = -u^*(-\tau)$ on
		$[-T^*/2, T^*/2]$, where $T^*$ is the minimal
		maneuver time~\cite{singh1989}.
		\item[\normalfont(P3)] The optimal control has at most
		three switches~\cite{pao1990}.
	\end{description}
	\smallskip
	By~(P1)--(P3), the optimal control is bang-bang with four intervals of durations $\{\alpha, \gamma, \gamma, \alpha\}$: a
	central \emph{synchronization phase} of
	duration~$T_s := 2\gamma \geqslant 0$ (which may vanish for
	specific step sizes~\cite{keppler2020}), flanked by two
	\emph{boost} bangs of duration~$\alpha$
	(Fig.~\ref{fig:admissible_control}).

	\begin{figure}
		\centering
		\footnotesize
		\def\svgwidth{1.0\linewidth}
		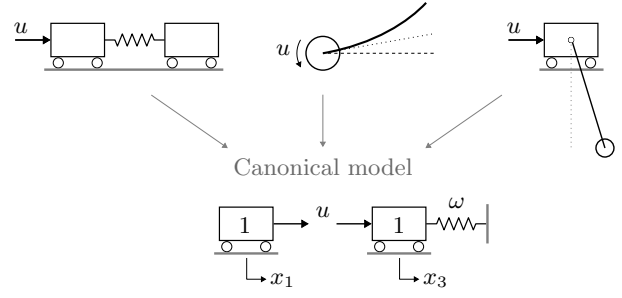
		\caption{(top) Three classical models (two-mass-spring, rigid body with flexible mode, overhead crane) reduce to a common canonical model (bottom): rigid body coupled to a harmonic oscillator through a shared input~$u$.}
		\label{fig:unification}
	\end{figure}
	\begin{figure}
		\centering
		\footnotesize
		\def\svgwidth{1.0\linewidth}
		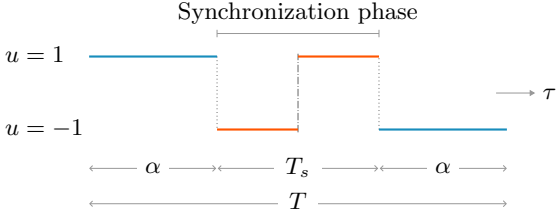
		\caption{Structure of the time-optimal control,	as implied by properties (P1)--(P3).}
		\label{fig:admissible_control}
	\end{figure}

\section{Phase-Space Geometry}\label{sec:geometry}
%

The state dynamics are block-diagonal, coupled only through the bang-bang input. The rigid body~(RB) follows parabolas in the $x_1 x_2$-plane, and the harmonic oscillator~(HO) clockwise circles in the $x_3 x_4$-plane. Switching times become parabolic heights and arc angles, the geometric language of the analysis that follows.


%

\subsection{Phase Portraits and Time-Geometry
	Connection}\label{sec:portraits}

For constant input $u(\tau) = \pm 1$, the solution
of~\eqref{eq:x_dyn} from an initial state
$\mathbf{x}(0) = (\bar x_1, \bar x_2,
\bar x_3, \bar x_4)^\top \in \mathbb{R}^4$ is
\begin{align}
	x_1(\tau) &= \pm\tfrac{1}{2}\tau^2
	+ \bar x_2\,\tau + \bar x_1,
	&
	x_2(\tau) &= \pm\tau + \bar x_2,
	\label{eq:rb-sol} \\
	x_3(\tau) &= C\,e^{i\tau}
	+ \bar C\,e^{-i\tau} \pm 1,
	&
	x_4(\tau) &= iC\,e^{i\tau}
	- i\bar C\,e^{-i\tau},
	\label{eq:ho-sol}
\end{align}
with $C \coloneqq \tfrac{1}{2}(\mp 1 + \bar x_3
- i\bar x_4)$
and $\bar C$ its complex conjugate. Throughout,
the upper (lower) sign corresponds to
$u \equiv +1$ ($u \equiv -1$). A solution segment
with $u \equiv +1$ is called a \emph{p-arc} and with
$u \equiv -1$ an
\emph{n-arc}
(terminology adapted from Bushaw~\cite{bushaw1952}).
Since the optimal control for Problem~1 is bang-bang with at
most three switches~\cite{pao1990}, every
minimum-time path is a concatenation of at most
four alternating p-~and n-arcs. A path beginning
with a p-arc (n-arc) is called a \emph{p-path}
(\emph{n-path}).


Eliminating~$\tau$ from~\eqref{eq:rb-sol} gives the
RB trajectories
\begin{equation*}
	x_1 = \pm\frac{x_2^2}{2}
	\mp \frac{\bar x_2^2}{2} + \bar x_1,
\end{equation*}
a family of parabolas passing through
$(\bar x_1, \bar x_2)$. Among these, only those intersecting the origin can serve as terminal segments of minimum-time paths. For the HO subsystem, define $z(\tau) \coloneqq x_3(\tau) + i\,x_4(\tau)$.
Then from~\eqref{eq:ho-sol},
\begin{equation}\label{eq:ho-circle}
	z(\tau) = 2\,\bar C\,e^{-i\tau} \pm 1,
\end{equation}
which describes clockwise circular arcs of unit angular
speed about the centers~$(\pm 1, 0)$ with radius
$r = \|z(0) \mp 1\|$. Arcs through the origin have radius $1$ and form the initial and terminal segments of minimum-time paths.

Each arc endows the traversal time~$\tau$ with two geometric
interpretations (Fig.~\ref{fig:time_geometry}). In the RB plane,
$\tau = |x_2(\tau) - x_2(0)|$ is the displacement along
$x_2$~\eqref{eq:rb-sol}. In the HO plane, $\tau$ is the clockwise
central angle from $z(0)$ to $z(\tau)$~\eqref{eq:ho-circle}, equal
to arc length on unit-radius arcs. 
Each bang-bang time interval
(Fig.~\ref{fig:admissible_control}) is thus a parabolic height in
the RB plane, and an arc angle in the HO plane.



\begin{figure}
	\centering
	\footnotesize
	\def\svgwidth{1.0\linewidth}
	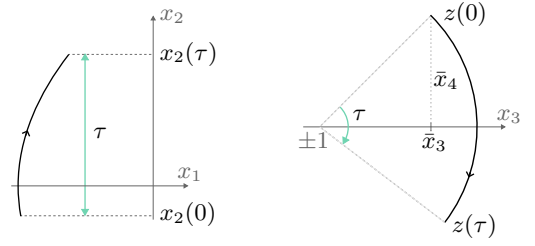
	\caption{Geometric interpretation of traversal
		time~$\tau$: vertical distance
		$|x_2(\tau)\!-\!x_2(0)|$ in the RB phase plane
		(left), clockwise swept angle from~$z(0)$
		to~$z(\tau)$ in the HO phase plane~(right).}
	\label{fig:time_geometry}
\end{figure}

Properties~(P1)--(P3) restrict every time-optimal
rest-to-rest maneuver to the palindromic
trajectories of Fig.~\ref{fig:switching-geom}. By
symmetry, take a p-path with initial
displacement~$\xb_1=-L$. In the RB plane, the
trajectory departs from~$(-L, 0)$ on a p-arc of
height~$\alpha$, then descends on an n-arc
by~$T_s/2=T/2-\alpha$ to the symmetry line $x_1 = -L/2$ at
the midpoint required by~(P1). In the HO plane, it
departs from the origin on a unit-radius p-arc
centered at~$(+1, 0)$, sweeps a clockwise
angle~$\alpha$ to the first switch, then continues
on a circle centered at~$(-1, 0)$ for an
angle~$T_s/2$ to the $x_4$-axis crossing required
by~(P1). In each plane, the second half mirrors the first by~(P2). This closes the trajectory at the origin.

\begin{figure}
	\centering
	\footnotesize
	\def\svgwidth{0.66\linewidth}
	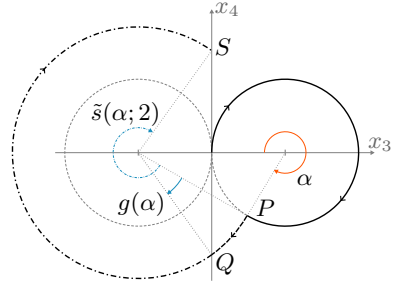
	\caption{Canonical HO trajectories (first half-maneuver). 
		The first arc (solid) ends at switching point $P$. The 
		second arc terminates on the $x_4$-axis at $Q$ ($k$ odd) 
		or $S$ ($k$ even). 
	}
	\label{fig:no_quicker_solution_2}
\end{figure}

\subsection{The Synchronization Angle}%
\label{sec:sync-geom}
We now determine the synchronization half-duration~$T_s/2$ for a given boost angle~$\alpha$: (P1) fixes $T_s/2$ up to an integer~$k$, the number of $x_4$-axis crossings.

The first bang ($u = +1$, duration~$\alpha$) traces a
unit-radius arc from the origin to the switching point
$P = (1 - \cos\alpha,\, \sin\alpha)$. After the switch, the
state moves along an arc centered at~$(-1, 0)$ with radius
\begin{equation}\label{eq:r-def}
	R = r(\alpha) \coloneqq \sqrt{5 - 4\cos\alpha},
\end{equation}
attaining its extrema $r = 1$ at $\alpha = 2n\pi$ and $r = 3$ at
$\alpha = (2n+1)\pi$.
%
By~(P1), $x_3(T/2) = 0$. The second bang ($u = -1$) starting from~$P$ produces $x_3(\tau) = -1 + (2 - \cos\alpha)\cos\tau + \sin\alpha\sin\tau$,  obtained by rotating $P$ clockwise by $\gamma$. Setting $x_3 = 0$ after a sweep of angle~$\gamma$ (Fig.~\ref{fig:switching-geom})  yields the \emph{midpoint-closure condition}
\begin{equation}\label{eq:midpoint} 
	(2 - \cos\alpha)\cos\gamma + \sin\alpha\,\sin\gamma = 1. 
\end{equation} 

\begin{figure*}
	\footnotesize
	\centering
	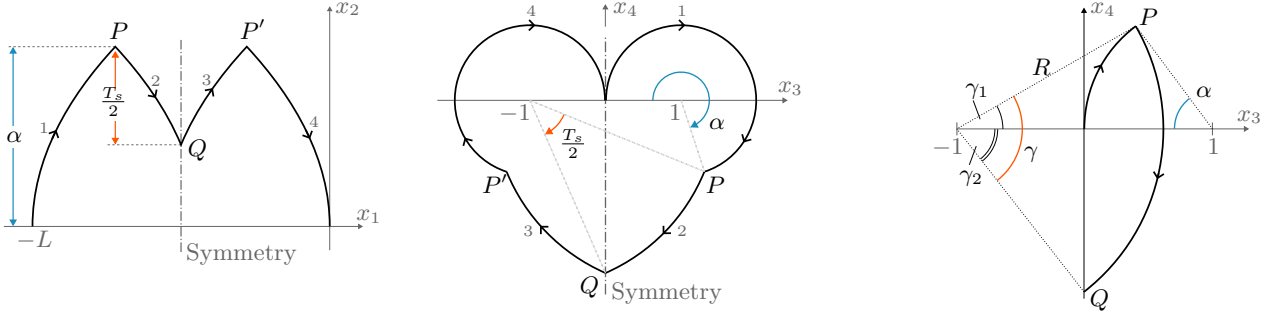
	\caption{Boost bang-duration~$\alpha$ and synchronization 
		half-duration~$T_s/2 = g(\alpha)$ in the rigid-body phase 
		plane (left, vertical segments) and harmonic-oscillator 
		phase plane (middle, arc angles). Right: decomposition 
		$g(\alpha) = \gamma_1 + \gamma_2$.}
	\label{fig:switching-geom}
\end{figure*}

The second arc reaches the $x_4$-axis at two points: $Q = (0, -\sqrt{R^2 -1})$ (first crossing reached clockwise from~$P$) and $S = (0, +\sqrt{R^2 - 1})$ (second); see Fig.~\ref{fig:no_quicker_solution_2}. 
At $\alpha = 2n\pi$ the radius is~$1$, $P$ returns to the origin, and the synchronization phase vanishes, giving the \emph{natural motions} of Theorem~\ref{thm:analytical_sol}.
%
%
The arc
$\widehat{PQ}$ decomposes at~$(-1, 0)$ into the two right-triangle
angles, see Fig.~\ref{fig:switching-geom}~(right),
\begin{equation}\label{eq:gamma_def}
	\gamma_1(\alpha) \coloneqq \arcsin\Bigl(\frac{\sin\alpha}{r(\alpha)}\Bigr), 
	\quad 
	\gamma_2(\alpha) \coloneqq \arccos\Bigl(\frac{1}{r(\alpha)}\Bigr).
\end{equation}
The arc from~$P$ to~$Q$ defines the \emph{synchronization angle}
\begin{equation}\label{eq:g_def}
	g(\alpha) \coloneqq \gamma_1(\alpha) + \gamma_2(\alpha)
	= \angle\widehat{PQ},
\end{equation}
with complementary arc from~$S$ to~$P$ with angle
\begin{equation}\label{eq:v-def}
	v(\alpha) \coloneqq \gamma_2(\alpha) - \gamma_1(\alpha)
	= \angle\widehat{SP}.
\end{equation}
Both map $\mathbb{R}_{\geqslant 0} \to \mathbb{R}_{\geqslant 0}$, are continuous, and $2\pi$-periodic.

Each $x_4$-axis crossing, together with any number
$j \in \mathbb{N}_{\geqslant 0}$ of additional
full revolutions, yields a distinct candidate
synchronization half-duration: $g(\alpha) + 2\pi j$
for arrivals at~$Q$ and
$[2\pi - v(\alpha)] + 2\pi j$ for arrivals at~$S$.
These two families are unified by a single
function indexed by a discrete \emph{mode
	index}~$k \in \mathbb{N}_{\geqslant 1}$ that
counts the $x_4$-axis crossings reached by the
synchronization arc:
\begin{equation}
	\label{eq:s_tilde}
	\tilde{s}(\alpha; k) \coloneqq
	\begin{cases}
		(k-1)\pi + g(\alpha),
		& k \text{ odd}, \\
		k\pi - v(\alpha),
		& k \text{ even}.
	\end{cases}
\end{equation}
The admissible synchronization half-duration $T_s/2$ takes
values exactly in $\{\tilde{s}(\alpha; k) : k \geqslant 1\}$
(Lemma~\ref{lem:canonical_structure}). For $\alpha \neq 2n\pi$,
the map $k \mapsto \tilde{s}(\alpha; k)$ is strictly increasing
(consecutive increments are $g + v$ and $2\pi - g - v$, both
positive since $r \in (1, 3]$), so $T_s/2$ determines $k$
uniquely. At natural motions $\alpha = 2n\pi$, $T_s/2 = 0$ for
all $k$ and the trajectory reduces to the rigid-body bang-bang
law.

\section{Candidate Optimal Controls}
\label{sec:canonical}

We introduce the family of \emph{canonical controls}: the bang-bang sequences consistent with properties~(P1)--(P3) and the switching geometry of Section~\ref{sec:geometry}. The family is indexed by the boost-bang duration~$\alpha$ and a mode index~$k$, whose optimal value $k^{*} = 1$ is identified in Section~\ref{sec:analytical_solution}. 
%
%
Within the type-1 subfamily ($k = 1$), the boost-bang duration~$\alpha$, maneuver time~$T$, and displacement~$L$ determine one another through strictly increasing bijections (Lemma~\ref{lem:f_w_properties}), and type-1 maximizes displacement at fixed maneuver time (Lemma~\ref{lem:dominance}). These are the two ingredients for establishing $k^{*} = 1$. Along the way, the midpoint-closure condition of Section~\ref{sec:geometry} is brought into a compact form,~\eqref{eq:closure_compact}--\eqref{eq:Ts_branch}, that powers the proofs here and recurs throughout Sections~\ref{sec:flexibility_penalty}--\ref{sec:design}.

%
%
%
%

\subsection{The Canonical Control Family}
\label{sec:canonical_family}
\begin{definition}\label{def:canonical}
	For $(\alpha,k)\in\mathbb{R}_{\geqslant 0}\times\mathbb{N}_{\geqslant 1}$ and $\bar{x}_1\neq 0$, with
	$\tilde{s}(\alpha;k)$ the synchronization half-duration~\eqref{eq:s_tilde}, set the
	\emph{maneuver time} and \emph{synchronization time}
	\begin{equation}\label{eq:f_tilde}
		T \coloneqq \tilde{f}(\alpha;k) \coloneqq 2\bigl(\alpha+\tilde{s}(\alpha;k)\bigr),
		\qquad T_s \coloneqq 2\tilde{s}(\alpha;k).
	\end{equation}
	The \emph{type-$k$ canonical control}
	$u(\,\cdot\,;\alpha,\bar{x}_1,k)\colon[-T/2,T/2]\to\{-1,+1\}$ is the bang--bang signal
	\begin{equation}\label{eq:canonical_control}
		u(\tau;\alpha,\bar{x}_1,k) = \operatorname{sgn}(\bar{x}_1)\cdot
		\begin{cases}
			-1, & \tau\in[-T/2,\,-T_s/2],\\
			+1, & \tau\in(-T_s/2,\,0],\\
			-1, & \tau\in(0,\,T_s/2],\\
			+1, & \tau\in(T_s/2,\,T/2].
		\end{cases}
	\end{equation}
\end{definition}

\begin{lemma}\label{lem:canonical_structure}
	Let $u^{*}$ be a time-optimal control of Problem~\ref{prob:top} with $\bar{x}_1 \neq 0$, and let
	$\alpha^{*} \geqslant 0$ be the duration of its first bang. There is a unique mode index
	$k^{*} \in \mathbb{N}_{\geqslant 1}$ such that
	\begin{equation}\label{eq:opt_control_law_canonical}
		u^{*}(\tau) = u(\tau;\, \alpha^{*}, \bar{x}_{1}, k^{*}),
		\qquad \tau \in [-T^{*}/2,\, T^{*}/2],
	\end{equation}
	with $T^{*} = \tilde{f}(\alpha^{*}; k^{*})$ and $T_{s}^{*} = 2\tilde{s}(\alpha^{*}; k^{*})$; moreover
	$\alpha^{*}$ satisfies the reachability condition
	\begin{equation}\label{eq:w_tilde}
		\tilde{w}(\alpha^{*}; k^{*}) = L,
		\quad
		\tilde{w}(\alpha;k) \coloneqq 2\alpha^2 - \bigl(\alpha - \tilde{s}(\alpha;k)\bigr)^2 .
	\end{equation}
\end{lemma}




\begin{proof}
	
	Assume $\bar x_1 < 0$. The case $\bar x_1 > 0$ follows by sign
	symmetry of~\eqref{eq:canonical_control}.
	By properties~(P2) and~(P3) and the geometric
	analysis of
	Sections~\ref{sec:portraits}--\ref{sec:sync-geom},
	the optimal control has the four-interval
	bang-bang form~\eqref{eq:canonical_control} for
	some $\alpha^{*} \geqslant 0$ and
	synchronization half-duration
	$T_{s}^{*}/2 = \tilde{s}(\alpha^{*};\, k^{*})$,
	with $k^{*} \in \mathbb{N}_{\geqslant 1}$
	uniquely determined by strict monotonicity of
	$\tilde{s}(\alpha^{*};\, \cdot)$ in~$k$
	(Section~\ref{sec:sync-geom}).
	Substitution into~\eqref{eq:f_tilde} and then
	into~\eqref{eq:canonical_control}
	yields~\eqref{eq:opt_control_law_canonical}.
	Integrating $\dot x_1 = x_2$ from $-T^{*}/2$ to $T^{*}/2$ with
	$x_1(-T^{*}/2) = \bar x_1$ gives 
	\[
	x_1(T^{*}/2) = \bar x_1 + \int_{-T^{*}/2}^{T^{*}/2}
	x_2(\tau)\,d\tau.
	\]
	The bang-bang form~\eqref{eq:canonical_control} makes $x_2$
	piecewise linear with slopes $\pm 1$, and rest-to-rest fixes
	$x_2(\pm T^{*}/2) = 0$. The integral evaluates to
	$-\operatorname{sgn}(\bar{x}_{1})\,
	\bigl[(T^{*})^{2}/4 - (T_{s}^{*})^{2}/2\bigr]
	= -\mathrm{sgn}(\bar x_1)\,\tilde w(\alpha^{*};
	k^{*}).$
	Setting	$x_1(T^{*}/2) = 0$ gives~\eqref{eq:w_tilde}.

\end{proof}

Under~\eqref{eq:w_tilde} the RB traverses $L$ in the same time the HO needs to close its orbit (cf.~Fig.~\ref{fig:switching-geom}).


\subsection{Properties of $f, g$ and $w$}
\label{sec:type1_properties}
Setting $k = 1$ in Definition~\ref{def:canonical}
gives $\tilde{s}(\alpha;1) = g(\alpha)$ and the
general maneuver-time~\eqref{eq:f_tilde} and
displacement~\eqref{eq:w_tilde} reduce to the type-1 maps $f, w \colon \mathbb{R}_{\geqslant 0} \to \mathbb{R}_{\geqslant 0}$ 
defined by
\begin{align}\label{eq:f_def}
	&f(\alpha) \coloneqq \tilde{f}(\alpha; 1)
	= 2\bigl(\alpha + g(\alpha)\bigr),
	\\
	\label{eq:w_def}
	&w(\alpha) \coloneqq \tilde{w}(\alpha; 1)
	= \frac{f(\alpha)^{2}}{4} - 2g(\alpha)^{2},
\end{align}
which satisfy $f(0) = w(0) = 0$. For a type-1 canonical control parameterized
by~$\alpha$, $T = f(\alpha)$ is the maneuver time and $L = w(\alpha)$ is the step size.



The central claim of this subsection, Lemma~\ref{lem:f_w_properties}
below, is that $f$ and $w$ are strictly increasing bijections. Combined with the step-size dominance of
Lemma~\ref{lem:dominance}, this establishes $k^{*} = 1$ and
thereby reduces the synthesis of
Section~\ref{sec:analytical_solution} to a scalar inversion. Both
proofs rest on a compact form of the closure, derived first. Along
the type-1 family, expanding $\cos(\alpha + g)$ shows that the
midpoint-closure condition~\eqref{eq:midpoint}, evaluated at the
synchronization angle $\gamma = g(\alpha)$, is equivalent to
\begin{equation}\label{eq:closure_compact}
	\cos(\alpha + g) \;=\; 2\cos g - 1,
\end{equation}
i.e.\ $\cos g = \cos^{2}\bigl((\alpha+g)/2\bigr)$, which reads in
time variables
\begin{equation}\label{eq:closure-cos}
	\cos(T_{s}/2) \;=\; \cos^{2}(T/4).
\end{equation}
Evaluated at the optimum, \eqref{eq:closure-cos} becomes the
closure equation of Theorem~\ref{thm:pythagorean}. Two consequences
follow, used throughout Sections~\ref{sec:flexibility_penalty}--\ref{sec:design}.

\begin{lemma}\label{lem:sync_bound}
	Along the type-1 family, $g(\alpha) \in [0, \pi/2]$, i.e.\ the
	synchronization time satisfies $T_{s} \in [0, \pi]$, with
	$g(\alpha) = 0$ exactly at $\alpha = 2n\pi$,
	$n \in \mathbb{N}_{0}$, and
	\begin{equation}\label{eq:Ts_branch}
		T_{s} = 4 \arcsin\Bigl(\tfrac{1}{\sqrt{2}}\,
		\bigl|\sin(T/4)\bigr|\Bigr).
	\end{equation}
	At these points the one-sided expansions are
	\begin{equation}\label{eq:g_expansion}
		g(2n\pi + \varepsilon)
		= (1 \pm \sqrt{2})\,\varepsilon + O(\varepsilon^{2}),
	\end{equation}
	with the upper sign for $\varepsilon > 0$ and the lower for
	$\varepsilon < 0$.
\end{lemma}

\begin{proof}
	Since the right side of~\eqref{eq:closure-cos} is nonnegative,
	$\cos g \geqslant 0$ for every $\alpha$. Every value in
	$(\pi/2,\, 3\pi/2)$ has negative cosine, so the continuous
	function $g$, which starts at $g(0) = 0$, can never enter this
	interval: by the intermediate value theorem it would first
	have to attain a value just above $\pi/2$, where
	$\cos g < 0$. Hence $g \in [0, \pi/2]$, i.e.\
	$T_{s} \in [0, \pi]$, and on this
	range~\eqref{eq:closure-cos} inverts uniquely: using
	$\cos(T_{s}/2) = 1 - 2\sin^{2}(T_{s}/4)$ and
	$\cos^{2}(T/4) = 1 - \sin^{2}(T/4)$, it
	gives $2\sin^{2}(T_{s}/4) = \sin^{2}(T/4)$,
	whence~\eqref{eq:Ts_branch}. Finally, $g = 0$
	in~\eqref{eq:closure_compact} forces $\cos\alpha = 1$. Conversely, at $\alpha = 2n\pi$ the closure reads
	$\cos g = 2\cos g - 1$, i.e.\ $\cos g = 1$, so $g = 0$ by
	$g \in [0, \pi/2]$.
	
	For the expansions, write $\alpha = 2m\pi + \varepsilon$ with
	$|\varepsilon|$ small. By continuity, $g \to 0$ as
	$\varepsilon \to 0$. Expanding both sides
	of~\eqref{eq:closure_compact} about
	$(\alpha, g) = (2m\pi, 0)$,
	$\cos(\varepsilon + g) = 1 - (\varepsilon+g)^{2}/2 + O^{4}$ and
	$2\cos g - 1 = 1 - g^{2} + O^{4}$, where $O^{4}$ collects
	fourth-order terms, so the closure reduces to
	$(\varepsilon + g)^{2} = 2g^{2} + O^{4}$, i.e.\
	$\varepsilon + g = \pm\sqrt{2}\,g$ to leading order. The sign of
	$\varepsilon$ selects the root. Since $g \geqslant 0$, the
	positive root gives $\varepsilon = (\sqrt{2}-1)\,g \geqslant 0$
	and the negative root gives
	$\varepsilon = -(\sqrt{2}+1)\,g \leqslant 0$. Solving each for
	$g$ yields~\eqref{eq:g_expansion}.
\end{proof}

The proofs below also use the following elementary chord--arc
inequality. It is stated separately because it recurs in
Lemma~\ref{lem:closure-derivative}.

\begin{lemma}\label{lem:sinc_gap}
	Let $T$ and $T_{s}$ be reals with $0 < T_{s} \leqslant \pi$ and
	$T > T_{s}$. Then
	\begin{equation}\label{eq:sinc_gap}
		T \sin(T_{s}/2) - T_{s} \sin(T/2) \;>\; 0 ,
	\end{equation}
	equivalently
	$\tfrac{2\sin(T_{s}/2)}{T_{s}} > \tfrac{2\sin(T/2)}{T}$, which is geometrically interpretable as the shorter arc is straighter.
\end{lemma}

\begin{proof}
	The two forms are related by division by $T T_{s} > 0$. We
	prove the second, i.e.\
	$\sin(T_{s}/2)/(T_{s}/2) > \sin(T/2)/(T/2)$.
	If $T/2 \leqslant \pi/2$, both arguments lie in $(0, \pi/2]$,
	where $x \mapsto \sin(x)/x$ is strictly decreasing (its
	derivative has the sign of $x\cos x - \sin x < 0$, by $\tan x > x$).
	If $T/2 > \pi/2$, then
	$\sin(T_{s}/2)/(T_{s}/2) \geqslant \sin(\pi/2)/(\pi/2) = 2/\pi$,
	whereas $\sin(T/2)/(T/2) \leqslant 2/T < 2/\pi$.
\end{proof}

\begin{lemma}\label{lem:f_w_properties}
	The type-1 maps $f$ and $w$ defined
	in~\eqref{eq:f_def}--\eqref{eq:w_def} are continuous, strictly
	increasing bijections of $\mathbb{R}_{\geqslant 0}$, with
	continuous, strictly increasing inverses. Moreover, $f$ is
	differentiable on
	$\mathbb{R}_{\geqslant 0} \setminus \{2\pi n\}_{n=0}^{\infty}$
	with strictly positive derivative. At the non-smooth points
	$f(2m\pi + \varepsilon) = f(2m\pi)
	+ 2(2 \pm \sqrt{2})\,\varepsilon + O(\varepsilon^{2})$, the
	upper sign for $\varepsilon \to 0^{+}$ and the lower for
	$\varepsilon \to 0^{-}$.
\end{lemma}

\begin{proof}
	We first verify the endpoint conditions and then
	address regularity.
	From \eqref{eq:f_def} with
	$g(0) = 0$ and $g \geqslant 0$, we have
	$f(0) = 0$ and $f(\alpha) \geqslant 2\alpha \to \infty$ as
	$\alpha \to \infty$. Likewise $w(0) = 0$ by~\eqref{eq:w_def}, and, since
	$g \leqslant \pi/2$ (Lemma~\ref{lem:sync_bound}),
	$w(\alpha) = (\alpha + g)^{2} - 2g^{2} \geqslant
	\alpha^{2} - \pi^{2}/2 \to \infty$.
	
	\emph{The map $f$.}\;
	$f$ is continuous by continuity of $g$
	(Section~\ref{sec:sync-geom}). On each open interval
	$\bigl(2m\pi,\, 2(m+1)\pi\bigr)$, $g$ is smooth: the arcsin
	argument in~\eqref{eq:gamma_def} satisfies
	$|\sin\alpha / r| < 1$ for all $\alpha$, since
	$\sin^{2}\alpha < r^{2}$ is equivalent to
	$(\cos\alpha - 2)^{2} > 0$, and the arccos argument satisfies
	$1/r < 1$ precisely for $r > 1$, i.e.\ away from the
	points $\alpha = 2m\pi$~\eqref{eq:r-def}. Both arguments are smooth in
	$\alpha$, hence so are $\gamma_{1}$, $\gamma_{2}$,
	and~$g$~\eqref{eq:g_def}. Moreover $g > 0$ there, by
	Lemma~\ref{lem:sync_bound}. With
	$g \in (0, \pi/2]$ (Lemma~\ref{lem:sync_bound}), substituting
	$\cos(\alpha+g) = 2\cos g - 1$ from the
	closure~\eqref{eq:closure_compact} yields the identity
	\begin{align*}
		(2\sin g)^{2} - \sin^{2}(\alpha + g) 
		&= 3 + \cos^{2}(\alpha + g) - 4\cos^{2} g \\
		&= 3 + (2\cos g - 1)^{2} - 4\cos^{2} g \\
		&= 4\,(1 - \cos g) \;>\; 0 .
	\end{align*}
	The left side factors as
	$\bigl(2\sin g - \sin(\alpha+g)\bigr)
	\bigl(2\sin g + \sin(\alpha+g)\bigr)$. Since the product is
	positive, both factors have the same sign, and their sum
	$4\sin g > 0$ makes both positive. In particular
	$D \coloneqq 2\sin g - \sin(\alpha+g) > 0$, and
	differentiating~\eqref{eq:closure_compact} gives $g' = \sin(\alpha+g)/D$, whence
	\begin{equation}\label{eq:f_prime}
		f'(\alpha) = 2\bigl(1 + g'(\alpha)\bigr)
		= {4\sin g}/{D} > 0.
	\end{equation}
	Since $f$ is continuous on each closed interval
	$\bigl[2m\pi,\, 2(m+1)\pi\bigr]$ and $f' > 0$ on its interior,
	the mean value theorem gives strict increase on each such
	interval, hence on $\mathbb{R}_{\geqslant 0}$. Combined with the
	endpoint conditions, $f$ is a bijection of
	$\mathbb{R}_{\geqslant 0}$ onto itself, so $f^{-1}$ is
	continuous and strictly increasing.
	At the non-smooth points, $f = 2(\alpha + g)$ and the one-sided
	expansions~\eqref{eq:g_expansion} of
	Lemma~\ref{lem:sync_bound} give
	$f(2m\pi + \varepsilon) = f(2m\pi)
	+ 2\bigl(1 + (1 \pm \sqrt{2})\bigr)\varepsilon
	+ O(\varepsilon^{2})$, which is the stated expansion.
	
	\emph{The map $w$.}\;
	Continuity of $w$ follows from that of $f$ and $g$. For
	monotonicity, parametrize the family by the maneuver time: by
	the first part of the proof, $f$ is a continuous, strictly
	increasing bijection
	with $f(2m\pi) = 4m\pi$ (as $g(2m\pi) = 0$,
	Lemma~\ref{lem:sync_bound}), so
	$\alpha \in \bigl(2m\pi,\, 2(m+1)\pi\bigr)$ corresponds to
	$T \in \bigl(4m\pi,\, 4(m+1)\pi\bigr)$,
	$m \in \mathbb{N}_{\geqslant 0}$. Since~\eqref{eq:closure-cos}
	has a unique solution $T_{s} \in [0, \pi]$ for each $T$, the
	branch~\eqref{eq:Ts_branch} evaluated at $T = f(\alpha)$ equals
	$2g(\alpha)$. Along the family, the quantities $T_{s}$
	and, by~\eqref{eq:w_def}, $L = T^{2}/4 - T_{s}^{2}/2$ are therefore
	determined by the maneuver time alone, and we may study their
	dependence on $T$. On each interval
	$\bigl(4m\pi,\, 4(m+1)\pi\bigr)$, $\sin(T/4) \neq 0$, so
	$T_{s}$ depends smoothly on $T$ through~\eqref{eq:Ts_branch},
	with $0 < T_{s} \leqslant \pi$.
	Moreover $T > T_{s}$, for $T \geqslant 2\pi$ this is immediate
	from $T_{s} \leqslant \pi$. For $0 < T < 2\pi$, both $T/4$ and
	$T_{s}/4$ lie in $[0, \pi/2)$, where the sine is strictly
	increasing, and~\eqref{eq:Ts_branch} gives
	$\sin(T/4) = \sqrt{2}\,\sin(T_{s}/4) > \sin(T_{s}/4)$, the
	inequality strict since $T_{s} > 0$, hence $T/4 > T_{s}/4$.
	Differentiating
	$\sin^{2}(T/4) = 2\sin^{2}(T_{s}/4)$ gives
	\begin{equation}\label{eq:dTs_dT}
		\frac{dT_{s}}{dT} = \frac{\sin(T/2)}{2\sin(T_{s}/2)} ,
	\end{equation}
	from where
	\begin{equation}\label{eq:dL_dT}
		\frac{dL}{dT}
		= \frac{T}{2} - T_{s}\,\frac{dT_{s}}{dT}
		= \frac{T\sin(T_{s}/2) - T_{s}\sin(T/2)}{2\sin(T_{s}/2)}
		\;>\; 0,
	\end{equation}
	where positivity of the numerator is
	Lemma~\ref{lem:sinc_gap}. Since $L$ depends continuously on $T$ on
	$\mathbb{R}_{\geqslant 0}$ and, by~\eqref{eq:dL_dT}, with
	strictly positive derivative on the interior of each interval
	$\bigl[4m\pi,\, 4(m+1)\pi\bigr]$, the mean value theorem shows
	that $L$ increases strictly with $T$ on each such closed
	interval, hence on $\mathbb{R}_{\geqslant 0}$.
	Since $T = f(\alpha)$ increases strictly with $\alpha$
	(shown above), so does $L = w(\alpha)$: the map $w$ is strictly
	increasing. Combined with the
	endpoint conditions, $w$ is a bijection of
	$\mathbb{R}_{\geqslant 0}$ onto itself, so $w^{-1}$ is
	continuous and strictly increasing.
\end{proof}
Type-1 maneuver time is thus a function of step size,
\begin{equation}\label{eq:h_def}
	h := f \circ w^{-1},
\end{equation}
a continuous, strictly increasing bijection of $\mathbb{R}_{\geq 0}$, as $h^{-1}$.




\subsection{Type-1 Dominance}
\label{sec:type1_dominance}

\begin{lemma}\label{lem:dominance}
	Type-1 dominates every type $k \geqslant 2$:
	\begin{enumerate}[label=(\roman*)]
		\item\label{itm:canonical-path-x4-intersec}
		It is strictly faster at fixed boost-bang duration: for
		all $\alpha > 0$,\; $f(\alpha) < \tilde{f}(\alpha; k)$.
		\item\label{itm:step_dominance}
		It achieves a strictly larger step size at fixed maneuver
		time: for all $T > 0$, if $\alpha_1 \coloneqq f^{-1}(T)$
		and $\alpha_k$ satisfies $\tilde{f}(\alpha_k; k) = T$, then $w(\alpha_1) > \tilde{w}(\alpha_k; k)$.
	\end{enumerate}
\end{lemma}

\begin{proof}
	\emph{Item~\textup{(i)}.}\;
	By~\eqref{eq:f_tilde} and~\eqref{eq:f_def},
	$f(\alpha) < \tilde{f}(\alpha; k)$ is equivalent to
	$\tilde{s}(\alpha;\, 1) < \tilde{s}(\alpha;\, k)$, which holds
	by strict monotonicity of $\tilde{s}(\alpha;\, \cdot)$ in~$k$
	(Section~\ref{sec:sync-geom}).
	
	\emph{Item~\textup{(ii)}.}\;
	We first show $\alpha_k < \alpha_1$, then prove that
	displacement is strictly increasing in~$\alpha$ at fixed~$T$.
	By Item~\textup{(i)} at $\alpha = \alpha_k$,
	$f(\alpha_k) < \tilde{f}(\alpha_k; k) = T = f(\alpha_1)$, and
	strict monotonicity of $f$
	(Lemma~\ref{lem:f_w_properties}) yields
	$\alpha_k < \alpha_1$.
	For any canonical control, regardless
	of type index $k$, with boost-bang duration
	$\alpha$ and maneuver time~$T$, the
	synchronization half-duration satisfies
	$\tilde{s} = T/2 - \alpha$
	(by~\eqref{eq:f_tilde}). Substituting
	into~\eqref{eq:w_tilde} eliminates both
	$\tilde{s}$ and $k$:
	\begin{equation*}
		\tilde{w}(\alpha; k)
		= 2\alpha^2 -
		\bigl(2\alpha - \tfrac{T}{2}\bigr)^2
		= \tfrac{T^2}{4} -
		2\bigl(\alpha - \tfrac{T}{2}\bigr)^2
		\eqqcolon W_T(\alpha).
	\end{equation*}
	Since $\tilde{s} \geq 0$ forces
	$\alpha \in (0, T/2]$, and $W_T$ is a downward parabola in
	$\alpha$ with vertex at $\alpha = T/2$, $W_T$ is strictly
	increasing on this interval.
	With $0 < \alpha_k < \alpha_1 \leq T/2$, this
	gives $\tilde{w}(\alpha_k; k) = W_T(\alpha_k)
	< W_T(\alpha_1) = w(\alpha_1)$.
\end{proof}


\section{The Analytical Solution}
\label{sec:analytical_solution}
We now assemble the results of Section~\ref{sec:canonical} into a complete analytical solution of Problem~\ref{prob:TOC}. Combining Lemma~\ref{lem:canonical_structure}, Lemma~\ref{lem:f_w_properties}, and Lemma~\ref{lem:dominance} yields $k^* = 1$ and reduces the synthesis to scalar inversion of~$w$. The resulting one-to-one correspondence between step size $L$, boost-bang duration~$\alpha$, and minimum maneuver time~$T$ is summarized in Fig.~\ref{fig:commutative}.
\begin{figure}
	\centering
%
%

\begin{tikzpicture}[
	node distance = 1.2cm and 2.1cm,
	varNode/.style={
		circle,
		draw=white,
		thin,
		minimum size=0.5cm,
		fill=white,
		font=\normalsize,
		outer sep=1pt
	},
	descNode/.style={
		font=\footnotesize,
		color=black!60,
		align=center
	},
	mapEdge/.style={
		->,
		>={Stealth[length=2mm]},
		thin,
		color=black
	},
	implicitEdge/.style={
		->,
		>={Stealth[length=2mm]},
		thin,
		densely dashed,
		color=black
	},
	inverseEdge/.style={
		->,
		>={Stealth[length=1.6mm]},
		very thin,
		dashed,
		color=black!60
	},
	every edge quotes/.style={
		auto,
		font=\small\itshape,
		inner sep=2pt
	}
	]
	
	%
	%
	\node[varNode]                        (alpha) {$\alpha$};
	\node[descNode, left=-0.2cm of alpha, xshift=-0.3cm] {First\\ bang-duration};
	\node[varNode, right=of alpha]        (L)     {$L$};
	\node[descNode, right=-0.0cm of L]            {Step size};
	\node[varNode, below=of alpha]        (Ts)    {$T_s$};
	\node[descNode, left=-0.0cm of Ts]           {Sync.\ time};
	\node[varNode, below=of L]            (T)     {$T$};
	\node[descNode, right=-0.0cm of T]            {Maneuver time};
	
	\draw[mapEdge] (alpha) to node[above]          {$w$}   (L);
	\draw[mapEdge] (alpha) to node[left]           {$2g$}  (Ts);
	\draw[mapEdge] (alpha) to node[above] {$f$}   (T);
	
	\draw[implicitEdge] (L) to node[right] {$h$} (T);
	
	\draw[mapEdge] (T) to node[above] {$\phi$} (Ts);
	
	
\end{tikzpicture}
	\vspace{-1.0em}
	\caption{Commutative diagram of the time-optimal synthesis.
		$\alpha$ determines $L$, $T$, and $T_{s}$. $\varphi$ relates
		$T$ to $T_{s}$ via the closure
		(Theorem~\ref{thm:pythagorean}), and $h = f\circ w^{-1}$
		solves $L \mapsto T$, the only arrow requiring inversion. Its
		inverse is closed form (Corollary~\ref{cor:kepler}).}
	\label{fig:commutative}
\end{figure}
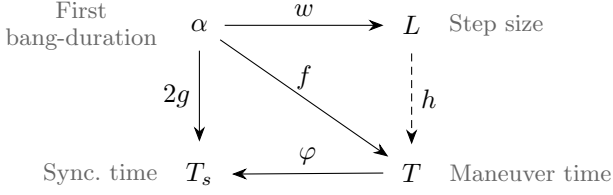



\label{sec:TOC_law}

\begin{theorem}\label{thm:analytical_sol}
	Consider Problem~\ref{prob:TOC} with step size
	$L := |\bar{x}_1|$, and let $f$, $g$, $w$, $h$ be as
	defined in~\eqref{eq:g_def}, \eqref{eq:f_def},
	\eqref{eq:w_def}, and~\eqref{eq:h_def}.
%
The unique time-optimal control is the type-1 canonical control ($k^* = 1$), and the following holds.
	\begin{enumerate}
		
		\item \emph{Synthesis via scalar inversion.}
		The optimal boost-bang duration is
		\mbox{$\alpha^* = w^{-1}(L)$}, the unique solution of
		\mbox{$w(\alpha) = L$}. The minimum maneuver
		and synchronization times are $T^* = f(\alpha^*)$
		and $T_s^* = 2g(\alpha^*)$. The optimal control
		law is
		
		\begin{equation}\label{eq:u_star}
			\hspace{-0.8em}u^*(\tau) = \operatorname{sgn}(\bar{x}_1)
			\cdot
			\begin{cases}
				-1, & \tau \in [-T^*\!/2,\; -T_s^*/2], \\
				+1, & \tau \in (-T_s^*/2,\; 0], \\
				-1, & \tau \in (0,\; T_s^*/2], \\
				+1, & \tau \in (T_s^*/2,\; T^*\!/2],
			\end{cases}
		\end{equation}
		%
		%
		%
		%
		or $u^*(\tau) = \operatorname{sgn}(\bar{x}_1)\,  
		\operatorname{sgn}\bigl(\tau (|\tau|-T_s^*/2)\bigr)$ a.e.\smallskip

		\item \emph{Minimum-time curve.} The pair $(L, T)$
		is achievable if and only if $T \geq h(L)$, with equality if and only if the maneuver is time-optimal.
		The boundary curve $T = h(L)$ admits the closed form  (Fig.~\ref{fig:perf_boundary})
		%
		\begin{equation}\label{eq:boundary_parametric}
			\bigl(L,\, T\bigr)
			= \bigl(w(\alpha),\, f(\alpha)\bigr),
			\qquad \alpha \geq 0.
		\end{equation}
		%
	\end{enumerate}
\end{theorem}

\begin{proof}
	
	We first prove \eqref{eq:boundary_parametric}, as item~\textup{(i)} uses it. Substituting $L = w(\alpha)$ into $T = h(L) = f(w^{-1}(L)) = f(\alpha)$ gives~\eqref{eq:boundary_parametric}.

	\emph{Item~\textup{(i)}:}	By Lemma~\ref{lem:canonical_structure}, the
	unique time-optimal control is a canonical control $u(\,\cdot\,;\, \alpha^*,\, \bar{x}_1,\, k^*)$ for some $\alpha^* \geq 0$ and $k^* \in \mathbb{N}_{\geq 1}$. Suppose, for contradiction, that $k^* \geq 2$.
	Let $\alpha_1 \coloneqq f^{-1}(T^*)$ be the first switching time
	of the type-1 control with maneuver time~$T^*$ ($f$ bijective,
	Lemma~\ref{lem:f_w_properties}). By
	Lemma~\ref{lem:dominance}\textup{(ii)}, this control attains
	a strictly larger step size than the optimal
	control: $\hat{L} \coloneqq w(\alpha_1)
	> \tilde{w}(\alpha^*;\, k^*) = L$,	where the final equality follows from the reachability condition (Lemma~\ref{lem:canonical_structure}\textup{(ii)}). 
%
%
	The type-1 control with displacement $L$ then has maneuver time $h(L) = f(w^{-1}(L))$. Since $h$ is strictly increasing and	$h(\hat{L}) = f(f^{-1}(T^*)) = T^*$, the ordering $L < \hat{L}$ yields $h(L) < h(\hat{L}) = T^*$ (cf.~Fig.~\ref{fig:perf_boundary}).
	This control completes the maneuver in time	$h(L) < T^*$, contradicting the minimality of~$T^*$. 
	Hence $k^{*} = 1$, and Lemma~\ref{lem:canonical_structure}(ii) gives $\alpha^{*} = w^{-1}(L)$, by bijectivity of $w$ (Lemma~\ref{lem:f_w_properties}). 
	Then $T^{*} = f(\alpha^{*})$ and $T_s^{*} = 2 g(\alpha^{*})$ follow from~\eqref{eq:f_def} and Definition~\ref{def:canonical}.
	The control law~\eqref{eq:u_star} is then the type-1 canonical control evaluated at	$\alpha^*$.
	
	\emph{Item~\textup{(ii)}:}\;
	By item~\textup{(i)}, $T^* = f(\alpha^*) = h(L)$ is the
	minimum time, so $(L, T)$ is achievable iff
	$T \geq h(L)$.
\end{proof}

\begin{corollary}
\label{cor:max_displacement}
Given a time budget $T > 0$, the maximum reachable displacement is
$L = h^{-1}(T) = w\bigl(f^{-1}(T)\bigr)$, where $h$ is the
minimum-time map of Theorem~\ref{thm:analytical_sol}.
\end{corollary}
\begin{proof}
	By Theorem~\ref{thm:analytical_sol} and strict monotonicity\\ 
	$T~\geq~h(L)\iff L \leq h^{-1}(T)= w \circ f^{-1}(T)$.
\end{proof}
%
%


\section{The Flexibility Penalty}
\label{sec:flexibility_penalty}

Flexibility imposes a temporal and a spatial cost. The rigid-body double integrator traverses a displacement~$L$ in minimum time~$T_r \coloneq 2\sqrt{L}$. The flexible system needs strictly more except at a discrete family of \emph{natural motions} (\emph{temporal penalty}), and for small~$L$ its excursions exceed~$L$ itself (\emph{spatial penalty}). We quantify both in closed form and refer to them collectively as the \emph{flexibility penalty}. 
Throughout, $T$ and $T_{s} = 2g(\alpha^{*})$ denote the optimal
maneuver and synchronization times, with $T-T_s=2\alpha\s>0$, $g$ from~\eqref{eq:g_def}
and $\alpha^{*} = w^{-1}(L)$ the boost-bang duration of the optimal
type-1 control (Theorem~\ref{thm:analytical_sol}). By Theorem~\ref{thm:analytical_sol}, $\alpha^*$ satisfies
the midpoint condition~\eqref{eq:midpoint} and the reachability condition \eqref{eq:w_tilde}.
\subsection{The Pythagorean Identity}
\label{sec:pythagorean}

\begin{theorem}
	\label{thm:pythagorean}
	For a common step size $L > 0$, the minimum times of the rigid
	body, $T_{r} = 2\sqrt{L}$, and of the flexible
	system~\eqref{eq:x_dyn}, $T$, satisfy
	\begin{alignat}{2}
		&\textit{Cost:}
		&\qquad T^{2}
		&\;=\; T_{r}^{\,2} + 2\,T_{s}^{\,2},
		\label{eq:pythag}\\[4pt]
		&\textit{Closure:}
		&\qquad \sin^{2}({T}/{4})
		&\;=\; 2\sin^{2}({T_{s}}/{4}).
		\label{eq:closure}
	\end{alignat}
	Given $L$, these two equations have a unique solution $(T, T_{s})$ with $T_{s} \in [0, \pi]$, which is therefore the optimal pair. The closure alone is independent of $L$.
\end{theorem}
\begin{proof}
	\emph{Cost.}\;
	Eq.~\eqref{eq:w_def} at $\alpha = \alpha^{*}$ reads
	$4L = T^{2} - 2T_{s}^{\,2}$; since $T_{r}^{\,2} = 4L$, this
	is~\eqref{eq:pythag}.
	
	\emph{Closure.}\;
	The midpoint condition~\eqref{eq:midpoint} rearranges to~\eqref{eq:closure-cos}. Solving for $T_{s}$
	gives~\eqref{eq:Ts_branch}, and squaring
	$\sin(T_{s}/4) = |\sin(T/4)|/\sqrt{2}$ gives~\eqref{eq:closure}.
	
	
	
	\emph{Uniqueness.}\;
	Since $T_{s} \in [0,\pi]$ puts $T_{s}/2$ in $[0, \pi/2]$, where
	$\cos$ is injective, the closure~\eqref{eq:closure-cos} gives
	$T_{s} = 2\arccos(\cos^{2}(T/4))$. Substituting
	into~\eqref{eq:pythag} gives~\eqref{eq:kepler}, whose left side
	is strictly increasing in $T$ by~\eqref{eq:dL_dT}. At most
	one pair $(T, T_{s})$ satisfies both equations, and by the first
	two parts the optimal pair does.
\end{proof}
%
%

\begin{figure}
	\centering
	\footnotesize
	\def\svgwidth{0.85\linewidth}
	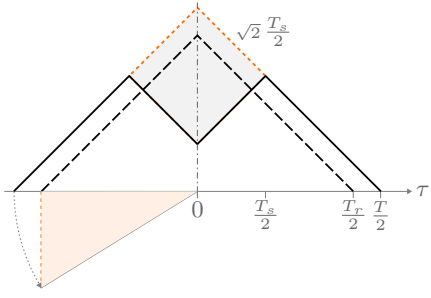
	\caption{Visual display of the Pythagorean
		identity~\eqref{eq:pythag}: flexible (solid) and
		rigid-body (dashed) velocity profiles enclose the same
		displacement~$L$. Equating areas yields the triangle.}
	\label{fig:pyth_ident}
\end{figure}
\begin{rem}
	The flexible velocity profile $x_{2}(\tau)$	encloses the displacement $L$ as an isosceles triangle of base $T$ minus a square of diagonal $T_{s}$, the geometric signature of the synchronization phase (Fig.~\ref{fig:pyth_ident}). The same $L$ is enclosed by the rigid-body
	triangle of base $T_{r} = 2\sqrt{L}$. Equating areas yields the right triangle with legs $T_{r}/2$	and $\sqrt{2}\,T_{s}/2$ and hypotenuse $T/2$; cf.~\eqref{eq:pythag}.
\end{rem}


\subsection{Time-Displacement Law}
Theorem~\ref{thm:pythagorean} gives a closed-form
time-displacement law. It defines $h^{-1}\colon T \mapsto L$
explicitly, so computing the minimum time $h(L)$ is a scalar
root-find on a strictly increasing function.
\begin{corollary}\label{cor:kepler}
	$T$ satisfies the time-displacement law
	\begin{equation}\label{eq:kepler}
		(T/2)^{2}
		- 2\arccos^{2}\!\big(\cos^{2}(T/4)\big) = L,
	\end{equation}
	the flexible analogue of the rigid-body law
	$(T_{r}/2)^{2} = L$. The subtracted term is
	$T_{s}^{2}/2 \in [0, \pi^{2}/2]$.
	%
\end{corollary}
\begin{proof}
	Substituting \eqref{eq:closure} into \eqref{eq:pythag} gives \eqref{eq:kepler}.
\end{proof}

\begin{figure}
	\centering
	\footnotesize
	\def\svgwidth{1.0\linewidth}
	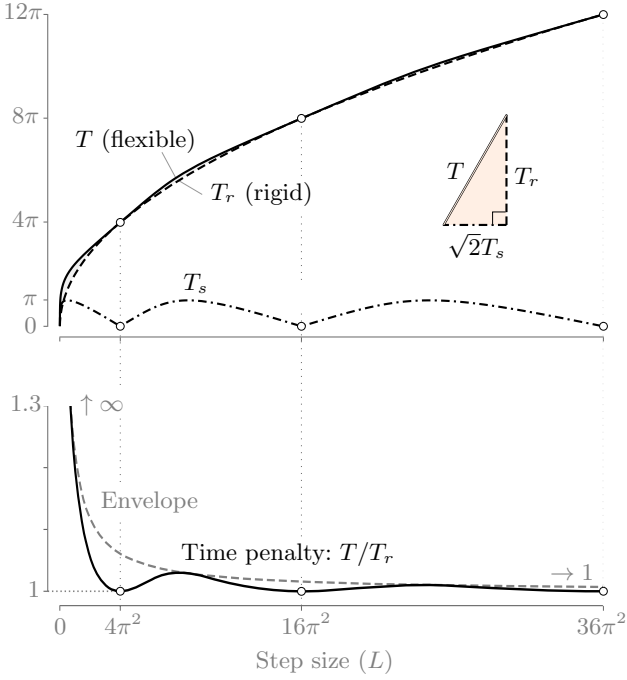
	\caption{
		Top: maneuver time $T$, rigid-body minimum $T_{r}$, and synchronization time $T_{s}$ vs. step size $L$. Minimum-time map $T = h(L)$ (Corollary~\ref{cor:kepler}) is the rigid-body parabola plus a bounded correction vanishing at natural motions (Corollary~\ref{cor:natural_anti}). Bottom: time penalty $T/T_{r}$ against envelope~\eqref{eq:envelope}.
		%
	}
	\label{fig:perf_boundary}
\end{figure}

\subsection{Natural Motions and Maneuver-Time Bounds}
\label{sec:bounds}
\begin{corollary}\label{cor:natural_anti}
	$T_{s}$ attains its bounds exactly when the maneuver spans a
	whole number of oscillator periods, $T = 2\pi m$,
	$m \in \mathbb{N}_{0}$, with $L = \pi^{2}m^{2} - T_{s}^{2}/2$:
	\begin{enumerate}
		\item even $m = 2n$: $T_{s} = 0$ and
		$L = L_{n}^{\textup{nat}} := (2n\pi)^{2}$, here $T = T_{r}$
		and~\eqref{eq:u_star} reduces to the rigid-body optimum, a single switch. These are the \emph{natural
			motions}~\cite{keppler2020}.
		\item odd $m = 2n+1$, $T_{s} = \pi$ and
		$L = L_{n}^{\textup{anti}} :=
		\pi^{2}\bigl[(2n+1)^{2} - \tfrac{1}{2}\bigr]$, the
		\emph{anti-resonances}.
	\end{enumerate}
	The families interlace:	$L_{n}^{\textup{nat}} < L_{n}^{\textup{anti}} < L_{n+1}^{\textup{nat}}$.
\end{corollary}

\begin{proof}
	By~\eqref{eq:closure-cos} and $T_{s} \in [0,\pi]$
	(Lemma~\ref{lem:sync_bound}), $T_{s} = 0$ at $T = 4\pi m$ and
	$T_{s} = \pi$ at $T = 2\pi(2m+1)$, $m \in \mathbb{N}_{0}$. Both
	read $T = 2\pi m$ with $m$ even and odd, respectively.
	Substituting into~\eqref{eq:pythag} gives
	$L = \pi^{2}m^{2} - T_{s}^{2}/2$, hence $L_{n}^{\textup{nat}}$ for
	$m = 2n$ and $L_{n}^{\textup{anti}}$ for $m = 2n+1$. At a natural
	motion $T = 4n\pi = 2\sqrt{L_{n}^{\textup{nat}}} = T_{r}$, and
	$T_{s} = 0$ collapses the inner intervals of~\eqref{eq:u_star},
	leaving the single sign change at $\tau = 0$. Consecutive $L$
	differ by at least $\pi^{2}(2m+1) - \pi^{2}/2 > 0$, giving the
	interlacing.
\end{proof}

\begin{corollary}\label{cor:bounds}
	For every $L > 0$,
	\begin{equation}\label{eq:envelope}
		2\sqrt{L}
		\;\leqslant\; T
		\;\leqslant\; 2\sqrt{L + \pi^{2}/2}\,;
	\end{equation}
	the bounds are the rigid-body times for	displacements $L$ and $L + \pi^{2}/2$, which are attained at the natural motions and anti-resonances of	Corollary~\ref{cor:natural_anti}, respectively (Fig.~\ref{fig:perf_boundary}).
\end{corollary}

\begin{proof}
	By Lemma~\ref{lem:sync_bound}, $T_{s} \in [0, \pi]$, so taking the square roots of \eqref{eq:pythag} yields~\eqref{eq:envelope}.
\end{proof}
%

\subsection{Small- and Large-Maneuver Limits}
\label{sec:scaling}
Flexibility is arbitrarily expensive for small maneuvers, yet
asymptotically free for large ones.
\begin{proposition}\label{prop:scaling}
	The time ratio $T/T_{r} \to \infty$ as $L \to 0$ and $T/T_{r} \to 1$ as
	$L \to \infty$, with leading-order scaling
	\begin{equation*}
		T = c\,L^{1/4} + O(L^{3/4})
		\quad\text{as } L \to 0,
		\qquad c := 384^{1/4}.
	\end{equation*}
	The quadruple integrator's time-optimal rest-to-rest law is $T_{4}(L) = c\,L^{1/4}$, with $c$ as above, so $T/T_{4} \to 1$ as $L \to 0$.
\end{proposition}
\begin{proof}
	Set $\theta := T/4$. By~\eqref{eq:Ts_branch},
	$T_{s} = 4\arcsin\bigl(\sin\theta/\sqrt{2}\bigr)$ for small
	$\theta > 0$, so
	$T_{s} = 2\sqrt{2}\,\theta - \tfrac{\sqrt{2}}{6}\theta^{3} + O(\theta^{5})$.
	Hence $T_{s}^{2} = 8\theta^{2} - \tfrac{4}{3}\theta^{4} + O(\theta^{6})$,
	and with $T^{2} = 16\theta^{2}$ the cost~\eqref{eq:pythag} gives
	$T_{r}^{2} = T^{2} - 2T_{s}^{2}
	= \tfrac{8}{3}\theta^{4} + O(\theta^{6})
	= T^{4}/96 + O(T^{6})$.
	Since $T_{r} = 2\sqrt{L}$, this is the stated expansion
	$T^{4} = 384\,L\bigl(1 + O(\sqrt{L})\bigr)$, thus
	$T/T_{r} = \tfrac{c}{2}L^{-1/4}(1 + o(1)) \to \infty$. For
	large $L$, the envelope~\eqref{eq:envelope} gives
	$1 \leqslant T/T_{r} \leqslant \sqrt{1 + \pi^{2}/(2L)} \to 1$.
	
	For the quadruple integrator $x^{(4)} = u$, $|u| \leqslant 1$,
	the time-optimal rest-to-rest control is bang-bang with three
	switches and, by uniqueness and the invariance of the problem
	under $t \mapsto T-t$, $u \mapsto -u$, satisfies
	$u(T-s) = -u(s)$. With $m_{k} := \int_{0}^{T}\! s^{k} u\, ds$,
	repeated integration from rest gives $x^{(3)}(T) = m_{0}$,
	$x^{(2)}(T) = Tm_{0} - m_{1}$,
	$x'(T) = \tfrac{1}{2}(T^{2}m_{0} - 2Tm_{1} + m_{2})$ and
	$x(T) = \tfrac{1}{6}(T^{3}m_{0} - 3T^{2}m_{1} + 3Tm_{2} - m_{3})$.
	The first terminal condition gives $m_{0} = 0$. The substitution
	$s \mapsto T-s$ gives $m_{k} = -\int_{0}^{T}(T-s)^{k} u\, ds$;
	expanding at
	$k = 2$ and using $m_{0} = 0$ yields $m_{2} = Tm_{1}$. The second
	and third conditions both reduce to $m_{1} = 0$, and the
	fourth becomes $L = -m_{3}/6$. Integrating the four bangs with
	switches at $a$, $T/2$, $T-a$ gives
	$m_{1} = \tfrac{1}{4}\bigl(T^{2} - 8Ta + 8a^{2}\bigr)$, whose
	root below $T/2$ is $a = \tfrac{T}{4}(2-\sqrt{2})$. With
	$a + (T-a) = T$ and $a(T-a) = T^{2}/8$ this yields
	$m_{3} = -T^{4}/64$, hence $L = T^{4}/384$, i.e.\
	$T_{4}(L) = c\,L^{1/4}$.
\end{proof}


\begin{figure}
	\centering
	\tiny
	\def\svgwidth{1.0\linewidth}
	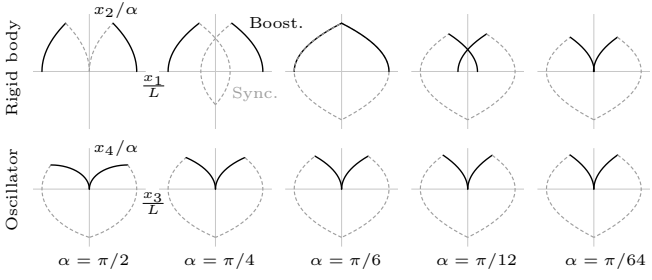
	\caption{Phase portraits for decreasing boost-bang
		duration~$\alpha$: rigid pair $(x_{1}, x_{2})$ (top, shifted
		by $-\bar{x}_{1}/2$ to center the maneuver) and
		oscillator pair $(x_{3}, x_{4})$ (bottom). $x_1,x_3$ are scaled by
		the step size~$L$, $x_2,x_4$ by~$\alpha$, and each panel is
		scaled individually, so agreement between the rows is one of
		shape, not magnitude. As $\alpha \to 0$ the two rows become
		congruent. Under each bang the rigid pair traces a parabola and
		the oscillator pair a unit circle of the same curvature,
		agreeing to third order. The fourth-order mismatch is the step
		size itself, the geometric view of $T \propto L^{1/4}$
		(Proposition~\ref{prop:scaling}).}
	\label{fig:phase_portraits}
\end{figure}

\subsection{The Spatial Penalty}
\label{sec:detour}


\begin{lemma}\label{lem:g_props}
	$g(\alpha) > \alpha$ on $(0, \pi/2)$.
\end{lemma}
\begin{proof}
	Suppose $g \leqslant \alpha$ at some $\alpha \in (0, \pi/2)$. There $0 < g < \pi/2$ by Lemma~\ref{lem:sync_bound}. Then
	$0 < 2g \leqslant \alpha + g < \pi$, so monotonicity of the
	cosine and the closure~\eqref{eq:closure_compact} give
	$2\cos g - 1 = \cos(\alpha + g) \leqslant \cos 2g
	= 2\cos^{2}g - 1$, i.e.\ $\cos g \leqslant \cos^{2} g$. But
	$\cos g \in (0, 1)$ for $g \in (0, \pi/2)$, so
	$\cos g > \cos^{2} g$, a contradiction.
\end{proof}

\begin{proposition}\label{prop:detour}
	The rigid-body velocity $x_2$ reverses sign within each synchronization
	interval if and only if $L < \pi^{2}/2$. Along the optimal
	maneuver,
	\begin{equation*}
		\lim_{L \to 0}\, \max_{\tau}
		{\lvert x_{i}(\tau) - x_{i}(0)\rvert}/{L}
		= \infty,
		\quad i \in \{1, 3\},
	\end{equation*}
	both ratios diverging as $L^{-1/2}$. For $L < \pi^{2}/36$ the
	rigid-body ratio ($i = 1$) exceeds one: the body moves farther
	from its start than the commanded displacement.
\end{proposition}
\begin{proof}
	The first bang ramps $|x_{2}|$ to $\alpha$ and the second bang reduces it by $g$, so the velocity reverses there iff
	$g > \alpha$, which by Lemma~\ref{lem:g_props} holds exactly on
	$(0, \pi/2)$, i.e.\ $L < w(\pi/2) = \pi^{2}/2$
	(Fig.~\ref{fig:phase_portraits}). With $\alpha = w^{-1}(L)$,
	integrating gives peak deviation
	$\max(L, \alpha^{2})$ of $x_{1}$. The orbit radius gives
	$r(\alpha) - 1$ for $x_{3}$. Both vanish as $\alpha^{2}$ while
	$L$ vanishes as $\alpha^{4}$
	(Proposition~\ref{prop:scaling}), so both ratios diverge as
	$L^{-1/2}$. By~\eqref{eq:w_def},
	$\max(L, \alpha^{2})/L > 1 \iff \alpha^{2} > w(\alpha)
	\iff g > 2\alpha$, which holds on $(0, \pi/6)$, i.e. 
	\mbox{$L < w(\pi/6) = \pi^{2}/36$}.
\end{proof}


\begin{corollary}\label{cor:no_overshoot}
	$x_{1} - x_{3}$ decreases monotonically from $L$ to $0$. Hence
	$\max_{\tau}|x_{1} - x_{3}| = L$. In the two-mass realization the
	load coordinate is a positive multiple of $x_{1} - x_{3}$
	(Fig.~\ref{fig:unification}, Appendix~\ref{app:canonical}) and
	therefore also non-increasing: for any step size and mass ratio,
	the load approaches its target without overshoot, while the driven coordinate's relative deviation diverges as $L^{-1/2}$ for every mass ratio.
\end{corollary}
\begin{proof}
	By the symmetry $(\mathbf{x}, u) \mapsto (-\mathbf{x}, -u)$ we may
	take $\bar{x}_{1} = -L$, so that $u^{*} = +1$ on the first boost,
	matching the geometry of Section~\ref{sec:geometry}. Put
	$d := x_{1} - x_{3}$ and $e := x_{4} - x_{2}$, so that
	$\dot{d} = -e$, $\dot{e} = -x_{3}$, and $e(\pm T/2) = 0$. On the
	first boost the oscillator runs from the origin on the unit circle of
	center $(1,0)$, so
	$x_{3} = 1 - \cos\theta \geqslant 0$. On the following
	synchronization interval it sweeps clockwise to the \emph{first}
	crossing of the $x_{4}$-axis, so $x_{3} > 0$ until $\tau = 0$.
	Hence $e$ decreases from $e(-T/2) = 0$ and is nonpositive on
	$[-T/2, 0]$. By~(P2) $x_{3}$ is odd about $\tau = 0$, so
	$e = -\int_{-T/2}^{\tau} x_{3}$ is even and $e \leqslant 0$
	throughout. Thus $\dot{d} = -e \geqslant 0$, and $d$ runs
	monotonically from $d(-T/2) = -L$ to $d(T/2) = 0$.
\end{proof}

\section{Design Guidelines}
\label{sec:design}


This section translates the results of Sections~\ref{sec:analytical_solution}--\ref{sec:flexibility_penalty} into the physical parameters of the three models (Fig.~\ref{fig:unification}, App.~\ref{app:canonical}).
Throughout, uppercase $T$ denotes a
minimum time and lowercase $t$ its physical counterpart, related
by $t = T/\omega$.
%
%
%
The minimal maneuver time \emph{factorizes} into the rigid-body
minimum and a dimensionless factor,
\begin{equation*}
	t^{*} = t_{r}R(\nu),\quad
	t_{r} \coloneqq \frac{T_{r}}{\omega} = 2\sqrt{\frac{M\delta}{b_{1}}},\;
	\nu \coloneqq \frac{\omega t_{r}}{4\pi} = 
	\sqrt{\frac{L}{L^{\text{nat}}_{1}}},
\end{equation*}
where $L = M\omega^{2}\delta/b_{1}$ is the normalized step size,
$\omega = \sqrt{k/\mu}$ the natural frequency,
$R \coloneqq t^{*}/t_{r}$ the \emph{time ratio}, and
$L^{\text{nat}}_{1} = (2\pi)^{2}$ the step size at the first
natural motion, and $\nu$ the dimensionless \textit{cycle count}. The five physical parameters thus enter only through $t_{r}$ and $\nu$.

Here $M$ is the total inertia, $\mu$ the reduced inertia, $k$ the modal stiffness, $\delta$ the commanded displacement, and $b_{1}$ the actuator authority. Table~\ref{tab:design} instantiates them per domain and Appendix~\ref{app:canonical} derives the reduction.

The rigid-body baseline $t_{r}$ does not depend on $\omega$, hence not on stiffness $k$. The flexible-mode input gain $b_{2}$ (one-bending-mode model only) does not affect $t^*$ since the rest-to-rest conditions are homogeneous in the oscillator state, so $b_{2}$ rescales the modal deflection and leaves every time unchanged.
\begin{table}
	\centering
	\footnotesize
	\caption{Native parameters.	
	}
	\label{tab:design}
	\setlength{\tabcolsep}{3pt}
	\begin{tabular*}{\linewidth}
		{@{\extracolsep{\fill}}lccc@{}}
		\toprule
		& Two-mass-spring & Flexible struct. & Overhead crane \\
		\midrule
		$\delta$
		& displacement & angle & displacement \\
		$M$
		& $m_m + m_l$ & $J_h + J_t$ & $m_t + m_p$ \\
		$\mu$
		& $m_m m_l / M$ & $J_h J_t / M$ & $m_t m_p / M$ \\
		$k$
		& $k$ & $k_b$ & $m_p g/\ell$ \\
		$b_1$
		& $F_{\max}$ & $\tau_{\max}$ & $F_{\max}$ \\
		\bottomrule
	\end{tabular*}
\end{table}

\subsection{Cycle Count and Performance Envelope}
\label{sec:cycle_count}
\label{sec:envelope_phys}
By definition, $\nu$ is the number of oscillator cycles per bang of the rigid-body maneuver over the same step. It depends on the step size and system parameters. 
At $\nu = n$ the design coincides with the $n$-th natural motion
$L^{\text{nat}}_{n} = (2n\pi)^{2}$ 	(Corollary~\ref{cor:natural_anti}), where the oscillator completes $n$
cycles per bang, no synchronization is required, and
$R = 1$.
Between integers, a synchronization phase ($t_{s} > 0$) follows the initial boost phase (Fig.~\ref{fig:admissible_control}). The performance
ratio is bracketed by
\begin{equation}\label{eq:envelope_R}
	1 \;\leqslant\; R(\nu)
	\;\leqslant\; \sqrt{1 + 1/(8\nu^{2})}.
\end{equation}
To evaluate $R$, solve the time-displacement law~\eqref{eq:kepler}
for $T$ within the envelope~\eqref{eq:envelope} with
$L = 4\pi^{2}\nu^{2}$, then $R = t^*/t_r=T/T_{r} = T/(4\pi\nu)$.
Both bounds are sharp, the lower at natural motions $\nu = n$, the upper at the anti-resonances
$\nu^{\text{anti}}_{n} = \tfrac{1}{2}\sqrt{(2n+1)^{2} - 1/2}$,
$n\in\mathbb{N}_0$ (Corollary~\ref{cor:natural_anti}), where
$t_{s} = \pi/\omega$. Synchronization is thus capped at half the
oscillator period, independent of $\delta$
(Lemma~\ref{lem:sync_bound}).

\subsection{Stiffness Sensitivity and the Flexibility Paradox}
\label{sec:sensitivity}

Stiffening does not monotonically reduce $t^*$.
\begin{lemma}\label{lem:closure-derivative}
	For $L \neq L_{n}^{\mathrm{nat}}$ ($n \in \mathbb{N}_{0}$), the
	synchronization time is a differentiable function of the step size, with
	\begin{equation}\label{eq:Ts-prime}
		T_{s}' \;:=\; \frac{dT_{s}}{dL}
		\;=\; \frac{\sin(T/2)}{T\sin(T_{s}/2) - T_{s}\sin(T/2)},
	\end{equation}
	where $T$ and $T_{s}$ are the optimal times at step size $L$.
\end{lemma}
\begin{proof}
	For $L \neq L_{n}^{\mathrm{nat}}$, $T$ lies in an open interval
	$(4n\pi,\, 4(n+1)\pi)$, where the proof of
	Lemma~\ref{lem:f_w_properties} establishes $dT_{s}/dT$ and
	$dL/dT > 0$ in~\eqref{eq:dTs_dT}--\eqref{eq:dL_dT}. By the inverse function
	theorem, $T$ and hence $T_{s}$ are differentiable in $L$, and the
	chain rule gives $T_{s}' = (dT_{s}/dT)\big/(dL/dT)$, which
	is~\eqref{eq:Ts-prime}, with denominator
	$T\sin(T_{s}/2) - T_{s}\sin(T/2) > 0$ by
	Lemma~\ref{lem:sinc_gap}; $0 < T_{s} \leqslant \pi$ by
	Lemma~\ref{lem:sync_bound}, and $T - T_{s} = 2\alpha^{*} > 0$.
\end{proof}
\begin{theorem}\label{thm:sign_reversal}
	For $\nu \neq n$, the sensitivity of the minimal maneuver time to the natural frequency is
	\begin{equation}\label{eq:sens_omega}
		\frac{\partial t^{*}}{\partial \omega}
		=
		\frac{2T_{s}}{T}\cdot
		\frac{2L T_{s}' - T_{s}}{\omega^{2}},
		\qquad T_{s}' \coloneqq \frac{dT_{s}}{dL}.
	\end{equation}
	At each natural motion $\nu = n \in \mathbb{N}_{\geqslant 1}$,
	$t^{*}$ is differentiable in $\omega$ with
	$\partial t^{*}/\partial \omega = 0$, and the derivative strictly
	reverses sign: $< 0$ as $\nu \nearrow n$, $> 0$ as
	$\nu \searrow n$. 
	A fractional miss $\eta \coloneqq (\nu-n)/n$ costs	$R - 1 = \eta^{2}/2 + O(\eta^{3})$.
%
\end{theorem}

\begin{proof}
	With $T_{r} = 2\sqrt L$, $t_{r} = 2\sqrt{L}/\omega = 2\sqrt{M\delta/b_{1}}$ is $\omega$-independent.
	By Lemma~\ref{lem:closure-derivative}, $T_{s}$ is differentiable
	in $L$ off natural motions, with $T_{s}' = dT_{s}/dL$
	from~\eqref{eq:Ts-prime}. Since $t_{s} = T_{s}/\omega$ depends on
	$\omega$ through both $L$ and $1/\omega$, $\partial t_{s}/\partial
	\omega = T_{s}'(\partial L/\partial \omega)/\omega -
	T_{s}/\omega^{2} = (2L\,T_{s}' - T_{s})/\omega^{2}$, using
	$\partial L/\partial \omega = 2L/\omega$. Differentiating the 
	Pythagorean identity \eqref{eq:pythag}, using $\partial t_{r}/\partial \omega = 0$ and $t_{s}/t^{*} = T_{s}/T$, gives~\eqref{eq:sens_omega}.
	
	Formula~\eqref{eq:sens_omega} holds for $\nu \neq n$, where
	$T_{s}'$ exists (Lemma~\ref{lem:closure-derivative}). The integers
	$\nu = n$ are the natural motions and need separate treatment.
	There $t_{s} = 0$, so $t^{*} = t_{r}$ and  $T = T_{r} = 4n\pi$ by \eqref{eq:pythag}. So only the local deviation matters. 
	Write $\Delta \coloneqq T - 4n\pi$. Since
	$\cos(T/4) = (-1)^{n}\cos(\Delta/4)$, the
	closure~\eqref{eq:closure-cos} reads
	$\cos(T_{s}/2) = \cos^{2}(\Delta/4)$, free of~$n$ and even
	in~$\Delta$. Expanding both sides,
	$1 - T_{s}^{2}/8 + O(T_{s}^{4})
	= 1 - \Delta^{2}/16 + O(\Delta^{4})$, so
	$T_{s}^{2} = \Delta^{2}/2 + O(\Delta^{4})$ and, since
	$T_{s} \geqslant 0$ (Lemma~\ref{lem:sync_bound}),
	$T_{s} = |\Delta|/\sqrt{2} + O(|\Delta|^{3})$.
	%
	%
	The
	cost~\eqref{eq:pythag} gives $T = T_{r} + O(T_{s}^{2})$, and
	$T_{r} = 4\pi\nu$, so $\Delta = 4\pi(\nu - n) + O((\nu-n)^{2})$
	for $n \geqslant 1$. Hence
	$T_{s} = 2\sqrt{2}\,\pi\,|\nu - n| + O((\nu-n)^{2})$, and
	$t_{s} = T_{s}/\omega$ vanishes linearly in $|\nu - n|$.
	Since $t^{*}$ depends on $t_{s}$ only through $t_{s}^{2}$, the
	kink is squared away: $t^{*} = t_{r} + t_{s}^{2}/t_{r} +
	O(t_{s}^{4})$, quadratic in $\nu - n$ by the preceding paragraph,
	with a strict minimum at $\nu = n$, so the derivative vanishes and reverses sign strictly. With $t_{r} = 4\pi\nu/\omega$,
	$R(\nu) - 1 = t_{s}^{2}/t_{r}^{2} = (\nu-n)^{2}/(2\nu^{2}) +
	O((\nu-n)^{3})$, and substituting $\nu = n(1+\eta)$ gives $R - 1 = \eta^{2}/2 + O(\eta^{3})$.
\end{proof}


Since $t_{r}$ is $\omega$-independent, $R$ and $t^{*}$ share their
critical points. This resolves the flexibility paradox: by
\eqref{eq:envelope_R} \mbox{$R \geqslant 1$}, so adding flexibility
to a rigid body never shortens a maneuver, and costs nothing
at the natural motions ($R = 1$). Yet $R$ is not monotone, so just
above an integer $\nu$ softening does shorten a maneuver.




\subsection{Performance Landscape}
\label{sec:landscape}


Fig.~\ref{fig:performance_landscape} plots the time overhead
$R-1$ against cycle count~$\nu$, visualizing
Theorem~\ref{thm:sign_reversal}. The slope reverses at
each natural motion, and once per basin near but not at the
anti-resonance.

\begin{figure*}
	\footnotesize
	\centering
	\def\svgwidth{0.83\linewidth}
	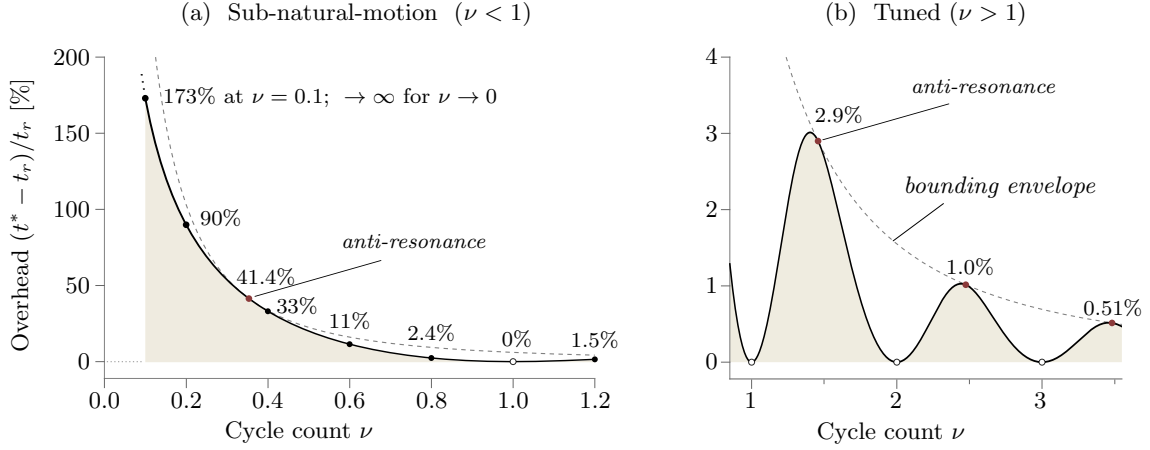
	\caption{Overhead $R(\nu) - 1$ versus cycle count
		$\nu = \omega t_r/(4\pi)$, with
		envelope~\eqref{eq:envelope_R} (dashed). Open circles: natural
		motions ($R = 1$). Red: anti-resonances $\nu^{\textup{anti}}_n$,
		where the envelope is attained.}
	\label{fig:performance_landscape}
\end{figure*}

\emph{The sub-natural-motion catastrophe ($\nu < 1$):} as
$\nu \to 0$ the time ratio $R$ grows as $\nu^{-1/2}$
(Prop.~\ref{prop:scaling}) and the excursion ratios of
Prop.~\ref{prop:detour} as $\nu^{-1}$, both without bound. Two
thresholds mark the regime: below the first anti-resonance
$\nu^{\textup{anti}}_0 = 1/(2\sqrt{2})$ the rigid-body velocity
reverses within each synchronization interval. Below $\nu = 1/12$
the deviation of the rigid coordinate from its start exceeds the
commanded displacement.

\emph{Flexibility is nearly free ($\nu \geqslant 1$):} the
envelope~\eqref{eq:envelope} bounds the overhead by
$\sqrt{1 + 1/(8\nu^{2})} - 1 \approx 1/(16\nu^{2})$, so
$R - 1 \leqslant \epsilon$ once $\nu \geqslant 1/(4\sqrt{\epsilon})$;
one percent needs $\nu \geqslant 2.5$.



\subsection{Design Rule}
\label{sec:design_rule}

\textit{Get to $\nu\geq 1$ at all costs:} the vibration period must not exceed the rigid-body bang, $2\pi/\omega \leqslant t_{r}/2$. Below $\nu = 1$ the penalty is unbounded and stiffness dominates the maneuver time.
%
%
Above it the penalty peaks at $3.01\%$. This supports the numerical finding of~\cite{marshall2023} that no universal heuristic exists for the ratio $T/T_{n}$ of maneuver time to natural period $T_{n} = 2\pi/\omega$. For the one-mode model the ratio is exact, $T/T_{n} = 2\nu R(\nu)$, and two periods ($\nu = 1$) already bring flexibility within $3.01\%$, while the standard heuristic requires $T/T_{n} \geqslant 10$, with 100 common in practice.
Targeting natural motions is forgiving: with $\omega_{n} = 4\pi n/t_{r}$ and $\omega = \rho\,\omega_{n}$, the penalty is
$R - 1 \approx (\rho-1)^{2}/2$, below $0.55\%$ for
$|\rho - 1| \leqslant 0.1$ and any $n$
(Theorem~\ref{thm:sign_reversal}). 
So $t^{*} \approx t_{r}$, and the maneuver time is dominated by inertia, displacement and actuator
authority, not by stiffness. Ten percent more actuator authority is worth $4.6\%$, more than perfect frequency matching can ever buy.

\subsection{Worked Example: Cable-Length Selection}
\label{sec:worked_example}


A post-Panamax STS container crane transports a payload of mass
$m_{p} = 40$\,t (container plus spreader) between bays separated
by $\delta = 25$\,m, on a trolley of mass $m_{t} = 40$\,t with
force budget $F_{\max} = 80$\,kN. Cable length $\ell$ is the only
design lever: the cable is the spring (Table~\ref{tab:design}). The rigid-body minimum time (the $\ell \to 0$ limit, trolley and payload moving as one mass) is $t_{r} = 2\sqrt{(m_{t}+m_{p})\delta/F_{\max}} = 10$\,s.
Tuning $\ell = 12.42$\,m places the system at the first natural motion
($R = 1$). The maneuver completes in
$t^{*} = t_{r} = 10.00$\,s. Shortening to $\ell = 5.85$\,m to
``go faster'' more than doubles the stiffness but lands at the
first anti-resonance ($\nu = \sqrt{8.5}/2$), where the
envelope~\eqref{eq:envelope_R} is sharp ($R = \sqrt{18/17}$), giving
$t^{*} \approx 10.29$\,s. The surprise is the sign, not the magnitude. 
	\section{Conclusions}
\label{sec:conclusions}
Time-optimal rest-to-rest control of a double integrator coupled to a
harmonic oscillator had resisted closed-form solution because the
switching times satisfy a transcendental boundary-value problem. This model underlies three classical benchmarks~(Fig.~\ref{fig:unification}). Recasting it in
the phase plane, where a single midpoint closure links the maneuver
and synchronization times, reduces the entire synthesis to inverting
one strictly increasing scalar function
(Theorem~\ref{thm:analytical_sol}).

The optimal time then splits into the rigid-body minimum and a bounded
synchronization cost (Theorem~\ref{thm:pythagorean}), which vanishes at
a discrete family of natural motions and saturates at the
anti-resonances (Corollary~\ref{cor:natural_anti}). For small steps the temporal penalty is unbounded, the minimum time approaching the quadruple integrator's exact law (Proposition~\ref{prop:scaling}), so one elastic mode costs as much as two additional integrators. The spatial penalty mirrors it, the trajectory's excursion growing without bound relative to the step, yet the load itself reaches its target without overshoot (Corollary~\ref{cor:no_overshoot}).

Sensitivity to stiffness is non-monotone (Theorem~\ref{thm:sign_reversal}): adding flexibility never speeds a maneuver, yet softening an already flexible structure sometimes does.
Stiffness is thus a parameter to tune rather than maximize, and the natural motions, where flexibility costs nothing and stiffness errors do not hurt to first order, are the design target.

The tractability rests on there being a single elastic mode, hence a
single closure. With several modes each contributes its own closure, and the optimal control gains switches accordingly; whether the problem still
collapses to a low-dimensional inversion, rather than the full
transcendental system, remains open.

	\begin{ack}                               
		The author thanks Prof.\ Alessandro\ De Luca for introducing him to time-optimal control of flexible joints.
	\end{ack}
	
	\appendix

\appendix

\section{System Dynamics and Transformations}
\label{app:canonical}

\renewcommand{\arraystretch}{1.1}
Each application-domain system $\dot{\mathbf{y}} =
\mathbf{A}\mathbf{y} + \mathbf{b}\,v$ with
$|v(t)| \leqslant 1$ maps to the canonical
form~\eqref{eq:x_dyn} via $\mathbf{y} =
\mathbf{T}\mathbf{x}$ and time scaling
$\tau = \omega t$, which yield
$
\dot{\mathbf{x}}\big|_{\tau}
= \tfrac{1}{\omega}\, \mathbf{T}^{-1}\mathbf{A}\mathbf{T}\,
\mathbf{x}
+ \tfrac{1}{\omega}\, \mathbf{T}^{-1}\mathbf{b}\, v,
$
matching~\eqref{eq:x_dyn} with $u(\tau) = v(t)$. The system-specific $\mathbf{A}$, $\mathbf{b}$, and $\mathbf{T}$ follow.
A physical step size $\delta$ maps to $\mathbf{x}_{0} = (\pm L, 0, 0, 0)^{\top}$ with $L = M\omega^2\delta/b_1$.

\emph{Two-mass-spring system and overhead
	crane.}
With masses $m_{m}, m_{l} > 0$, stiffness $k > 0$,
and input gain $b_{1}$,
\begin{equation}\label{eq:two-mass-spring_sys}
	\mathbf{A} =
	\begin{bNiceMatrix}[columns-width=0.8em,margin=0.4em]
		\mathbf{0} & \mathbf{I}_{2} \\
		\mathbf{K} & \mathbf{0}
	\end{bNiceMatrix},
	\quad
	\mathbf{K} =
	\begin{bNiceMatrix}[columns-width=0.8em,margin=0.4em]
		-k/m_{m} & k/m_{m} \\
		k/m_{l} & -k/m_{l}
	\end{bNiceMatrix},
\end{equation}
$\mathbf{b} = (0, 0, b_{1}/m_{m}, 0)^{\top}$, with
$\mathbf{y} = (q_{m}, q_{l}, \dot{q}_{m},
\dot{q}_{l})^{\top}$. With $M \coloneqq m_{m} +
m_{l}$ and $\omega \coloneqq \sqrt{kM/(m_{m}
	m_{l})}$,
\begin{equation*}
	\mathbf{T} = \frac{b_{1}}{M\omega^{2}}\,
	\begin{bNiceMatrix}[columns-width=0.8em,margin=0.4em]
		1 & 0 & m_l/m_m & 0 \\
		1 & 0 & -1 & 0 \\
		0 & \omega & 0 & \omega m_l/m_m \\
		0 & \omega & 0 & -\omega
	\end{bNiceMatrix}.
\end{equation*}
In the small-angle limit, the overhead
crane (trolley mass $m_t$, payload $m_p$, cable
length $\ell$) reduces to this system under
$m_m \to m_t$, $m_l \to m_p$, $k \to k_{\text{eff}}
\coloneqq m_p g/\ell$.

\emph{One-bending-mode model~\cite{ben-asher1987}.}
With hub inertia $J_h$, tip inertia $J_t$, total inertia $M = J_h + J_t$,
reduced inertia $\mu = J_h J_t/M$, bending stiffness $k_b$,
and natural frequency $\omega =\!\sqrt{k_b/\mu}$,
\begin{equation}\label{eq:one-bending-mode_sys}
	\mathbf{A} =
	\begin{bNiceMatrix}[columns-width=0.4em,margin=0.0em]
		\mathbf{A}_{r} & \mathbf{0} \\ \mathbf{0} & \mathbf{A}_{f}
	\end{bNiceMatrix},
	\quad
	\mathbf{A}_{r} =
	\begin{bNiceMatrix}[columns-width=0.4em,margin=0.05em]
		0 & 1 \\ 0 & 0
	\end{bNiceMatrix},
	\;
	\mathbf{A}_{f} =
	\begin{bNiceMatrix}[columns-width=0.4em,margin=0.05em]
		0 & 1 \\ -\omega^{2} & 0
	\end{bNiceMatrix},
\end{equation}
$\mathbf{b} = (0,\, b_{1}/M,\, 0,\, b_{2})^{\top}$ with
$\mathbf{y} = (y_{r}, \dot{y}_{r}, y_{f}, \dot{y}_{f})^{\top}$
stacking rigid-body and flexible coordinates, and
$\mathbf{T} = \omega^{-2}\operatorname{diag}\!\big(b_{1}/M,\; b_{1}\omega/M,\; b_{2},\; b_{2}\omega\big)$. $b_2$ is the flexible-mode input gain.

	\bibliographystyle{plain}
	
	\bibliography{MyLibrary-clean}

\end{document}